\documentclass[11pt]{article}
\pdfoutput=1

\usepackage[utf8]{inputenc}
\usepackage[T1]{fontenc}
\usepackage[letterpaper,margin=1in]{geometry}
\usepackage{microtype}
\usepackage[usenames,dvipsnames]{xcolor}
\usepackage{amsmath,amssymb,amsthm,mathtools,mathrsfs}
\usepackage{enumitem}
\usepackage{algpseudocode}
\algrenewcommand\algorithmicrequire{\textbf{Input:}}
\algrenewcommand\algorithmicensure{\textbf{Output:}}
\usepackage{booktabs,array,needspace,longtable,calc}
\usepackage[numbers,sort&compress]{natbib}
\usepackage{hyperref}
\usepackage[nameinlink,noabbrev]{cleveref}
\crefname{theorem}{Theorem}{Theorems}
\crefname{lemma}{Lemma}{Lemmas}
\crefname{corollary}{Corollary}{Corollaries}
\crefname{proposition}{Proposition}{Propositions}
\crefname{definition}{Definition}{Definitions}
\crefname{algorithm}{Algorithm}{Algorithms}
\crefname{section}{Section}{Sections}
\crefname{subsection}{Section}{Sections}
\crefname{subsubsection}{Section}{Sections}
\crefname{appendix}{Appendix}{Appendices}
\crefname{table}{Table}{Tables}

\hypersetup{
  pdftitle={Kadison--Singer partitions and Bilu--Linial graph signings in polynomial time},
  pdfauthor={Ali Jadbabaie, Amin Saberi, Suvrit Sra},
  pdfsubject={Kadison--Singer partitions and Bilu--Linial-scale graph signings},
  pdfkeywords={Kadison-Singer, Bilu-Linial, matrix discrepancy, deterministic rounding, spectral potential, graph signings, randomized repair}
}

\usepackage{ss}

\definecolor{cdarkblue}{RGB}{30,75,170}
\definecolor{cdarkred}{RGB}{180,0,0}
\definecolor{cdarkgreen}{RGB}{0,130,0}

\setlist[itemize]{topsep=0.35em,itemsep=0.2em,leftmargin=2em}
\setlist[enumerate]{topsep=0.35em,itemsep=0.2em,leftmargin=2.2em}
\allowdisplaybreaks[1]
\newcommand{\R}{\mathbb R}

\newcommand{\E}{\mathbb E}

\DeclareMathOperator{\Tr}{tr}
\DeclareMathOperator{\Sym}{Sym}

\DeclareMathOperator{\diag}{diag}
\newcommand{\norm}[1]{\left\|#1\right\|}

\DeclareMathOperator{\adj}{adj}
\DeclareMathOperator{\ran}{range}
\newcommand{\HH}{\mathbb H}
\newcommand{\Eop}{\mathcal E}
\newcommand{\Lop}{\mathcal L}
\newcommand{\Bcal}{\mathcal B}
\newcommand{\ip}[2]{\langle#1,#2\rangle}
\newcommand{\Cstar}{C_*}
\newcommand{\Czero}{C_0}

\newtheorem{theorem}{Theorem}[section]
\newtheorem{proposition}[theorem]{Proposition}
\newtheorem{lemma}[theorem]{Lemma}
\newtheorem{corollary}[theorem]{Corollary}
\theoremstyle{definition}

\newtheorem{algorithm}[theorem]{Algorithm}
\theoremstyle{remark}
\newtheorem{remark}[theorem]{Remark}

\numberwithin{equation}{section}

\makeatletter
\g@addto@macro\@maketitle{\vskip 0.10in\noindent\normalsize\itshape\@date}
\makeatother

\title{Kadison--Singer partitions and Bilu--Linial graph signings in polynomial time}
\author{\name Ali Jadbabaie \email{jadbabai@mit.edu}\\
\addr Massachusetts Institute of Technology\\[2pt]
\name Amin Saberi \email{saberi@stanford.edu}\\
\addr Stanford University\\[2pt]
\name Suvrit Sra \email{s.sra@tum.de}\\
\addr Technical University of Munich\\[2pt]
}
\date{19 September 2026}

\begin{document}
\maketitle

\begin{abstract}
\small
We prove two main algorithmic results in spectral discrepancy.
First, we give a deterministic polynomial-time rounding theorem for rational
positive semidefinite matrices of arbitrary rank. The algorithm starts from any rational fractional signing and assigns one
sign per original matrix. Its discrepancy is less than $3.37\,\|\sum_i \Tr(A_i)A_i\|^{1/2}$.
This yields Kadison--Singer half-partitions with error below
$1.69\sqrt{\varepsilon}$, as well as deterministic graph signings that control
signed adjacency and signed degrees simultaneously.  The proof builds on the spectral-potential method of Ezeunala and Jiang
(2026) and introduces a new way to choose rounding directions.
We prove polynomial bit complexity for the rounding procedure.

Second, we give a Las Vegas algorithm for the Bilu--Linial signing problem on
an arbitrary prescribed graph.  If $G$ has $n$ vertices and maximum degree $\Delta\ge3$, the algorithm
terminates almost surely. It uses fewer than $100n^3$ insertion attempts in
expectation and returns a signing with $\|A_s\|<2\sqrt{2(\Delta-1)}$.
For bipartite graphs its one-sided form gives the sharp universal bound
$\|A_s\|<2\sqrt{\Delta-1}$.  The algorithm builds the signing by inserting vertices and recursively
deleting and restoring neighbors after rejected insertions. 
In the analysis, the $\sqrt2$ gap to the Bilu--Linial conjecture comes from
a factor of two in the bound for vertex deletions in the two-sided case.  On a $d$-regular bipartite Ramanujan base the same signing produces a
Ramanujan $2$-lift of that prescribed base.
\end{abstract}

\section{Introduction}\label{sec:intro}

Spectral discrepancy asks whether discrete signs can balance a matrix-valued
sum with a bound independent of the ambient dimension. We consider two forms
of this problem. In Kadison--Singer discrepancy, one signs positive
semidefinite summands to balance every direction simultaneously. In graph
signing, one signs the edges of a prescribed graph to make the spectral radius
of its signed adjacency matrix small. We give an algorithm for each problem,
using different methods.

Our first main result is a deterministic algorithm for matrix discrepancy
with higher-rank inputs. At rank one, Weaver's formulation~\cite{Weaver} connected vector discrepancy
to the Kadison--Singer problem, and Marcus, Spielman, and
Srivastava~\cite{MSS2} proved the required rank-one existence theorem.  In the
matrix-variance scale, Kyng, Luh, and Song~\cite{KLS} obtained coefficient
$4$, and Xie, Xu, and Zhu~\cite{XXZ} improved the existential coefficient to
$3$.  Ezeunala and Jiang~\cite{EJ} gave a deterministic polynomial-time
rank-one rounding algorithm with coefficient $13$.  In contemporaneous work, Song and Yue~\cite{SongYue} obtain coefficient
$3.3443$ for rank-one inputs in deterministic polynomial time.
Kathuria~\cite{KathuriaKS} gives a different algorithm for
Weaver's problem in the real-arithmetic model.

Our coefficient below is therefore not the best rank-one constant. But our
theorem addresses the more general setting of matrices of arbitrary rank signing. It also allows arbitrary rational fractional starts, and we prove polynomial bit complexity in the trace scale.

To handle higher-rank inputs, we adapt the potential of
Ezeunala and Jiang~\cite{EJ}. The potential is an upper bound on the discrepancy,
obtained by minimizing over two auxiliary positive definite matrices and a
scalar. We introduce a way to choose local rounding directions that decrease
this potential. A completion-of-squares identity reduces the required
second-derivative estimate to a fixed two-dimensional scalar inequality. Verifying this inequality is the only
computer-assisted step. The verification uses exact integer arithmetic; the
certificate and a program that checks it will be made available on
GitHub~\cite{certificate}.

To state the rounding theorem, take positive semidefinite matrices
$A_1,\ldots,A_N$ and a fractional signing $x^0\in[-1,1]^N$.  The relevant
scale (``variance'') is
\begin{equation}\label{eq:scale}
 V=\Big\|\sum_{i=1}^N\Tr(A_i)A_i\Big\|.
\end{equation}
Here $\|\cdot\|$ denotes the operator norm.  The trace scale equals
$\|\sum_iA_i^2\|$ at rank one and dominates it in general, since
$A_i^2\preceq\Tr(A_i)A_i$.

\begin{theorem}[Deterministic PSD rounding]\label{thm:main}
Let $A_1,\ldots,A_N$ be rational Hermitian positive semidefinite $d\times d$
matrices of arbitrary rank, and let $x^0\in[-1,1]^N$ be rational. If $V>0$, a
deterministic algorithm returns $s\in\{-1,1\}^N$ such that
\begin{equation}\label{eq:main}
 \Big\|\sum_i(s_i-x_i^0)A_i\Big\|<3.37\sqrt V.
\end{equation}
Its running time is polynomial in $N$, $d$, and the input bit length. Each
original matrix receives one sign. If $V=0$, every input matrix is zero.
\end{theorem}


Applying the theorem to $\sum_i A_i=T$ with
$\Tr(A_i)\le\varepsilon$ gives two parts, each within
$1.69\sqrt{\varepsilon\norm T}$ of $T/2$.  This covers Kadison--Singer and
higher-rank partitions without requiring their total to be the identity.
Repeated halving then yields multi-way partitions and small-diagonal paving.
For multi-way partitions, related work includes the asymptotically optimal
multi-paving theorem of Ravichandran and Srivastava~\cite{RS}. For balanced
signings more broadly, Boolean small-ball methods~\cite{BSB} provide another
analytic approach. Our discrepancy bound comes from controlling the potential throughout
rounding. Each fractional coordinate has a positive weight that vanishes when the coordinate becomes a sign. Two coupled matrix inequalities use these weights to bound the largest and smallest eigenvalues
of the signed sum. Rounding a coordinate removes its weighted terms from both
inequalities. If they remain feasible with the same auxiliary variables, the
potential cannot increase. If every coordinate fails this test, we instead
decrease the potential by a small change in the fractional signing.

To choose this change, we linearize the two inequalities at the current minimizer, holding
their common scalar bound fixed. The resulting first-order changes in the
auxiliary matrices are their \emph{responses}. We construct a finite family
of coordinate directions whose responses are orthogonal in a suitable inner
product, and call it the \emph{response frame}. A completion-of-squares
identity bounds the potential's second derivative averaged over these
directions, retaining the interaction between the two responses. For
higher-rank inputs, trace inequalities for positive matrices replace the
rank-one identities. The resulting two-dimensional inequality holds with a
strict margin.

This extension to higher rank gives graph signings by encoding each edge
as one rank-two PSD matrix. Each matrix receives one sign, so the theorem
controls signed adjacency and signed-degree matrices simultaneously. The construction also
allows weights and fractional starts; Section~\ref{sec:applications} gives
the details. These bounds are of order $\sqrt\Delta$, but their constants
remain far from the Bilu--Linial bound. This motivates our second main result,
which treats signing a prescribed graph directly.

For this prescribed-graph problem, our second main result is a Las Vegas
algorithm.  Bilu and Linial~\cite{BL} related edge signings to graph lifts and
conjectured that every $\Delta$-regular graph admits a signing with
$\|A_s\|\le2\sqrt{\Delta-1}$.  Marcus, Spielman, and Srivastava~\cite{MSS1} proved the sharp one-sided
existence theorem, which bounds the largest eigenvalue. For bipartite graphs,
spectral symmetry gives the same bound on the magnitude of the smallest
eigenvalue.  For an
arbitrary prescribed graph of maximum degree at most $\Delta$, our randomized
repair algorithm finds in expected polynomial time a signing satisfying
\[
        \|A_s\|<2\sqrt{2(\Delta-1)}.
\]
For a prescribed bipartite graph it reaches the sharp universal bound
$\|A_s\|<2\sqrt{\Delta-1}$.  The expected number of insertion attempts is less
than $100n^3$, with an explicit exponential tail bound.

The algorithm obtains these bounds by inserting vertices one at a time. Determinants of matrices
encoding the spectral constraints determine the probabilities of accepting an
insertion or deleting a neighbor. After a deletion, the algorithm restores
that neighbor recursively before retrying the insertion. It does not search
over interlacing polynomials. Section~\ref{sec:repair} gives the theorem and
proof. The $\sqrt2$ loss in the general case comes from a factor of two in the
bound for deletions when both ends of the spectrum are constrained.  For a prescribed $d$-regular bipartite Ramanujan
graph, the output signing specifies a Ramanujan $2$-lift of that base.

The prescribed-base signing problem is distinct from constructing a
Ramanujan graph of a requested size and degree.  Cohen~\cite{Cohen} gives a deterministic polynomial-time algorithm for the
construction of Marcus, Spielman, and Srivastava~\cite{MSS4}, which uses
matchings and allows all sizes. His algorithm does not take an arbitrary
prescribed bipartite graph and sign its existing edges.

\subsection{Contributions}
We summarize the two algorithms, their analyses, and their
consequences below.
\begin{itemize}
\item \textbf{Main result I: rounding intact higher-rank summands.}
Theorem~\ref{thm:main} gives deterministic rounding in the trace scale from
any rational fractional start. It assigns one sign to each original PSD
matrix, regardless of rank.  No isotropic normalization or vector factorization is
required.  The proof establishes polynomial bit complexity.

\item \textbf{Main result II: randomized repair at the Bilu--Linial scale.}
Theorem~\ref{thm:repair} gives a Las Vegas algorithm that runs in expected
polynomial time. It signs every graph of maximum degree $\Delta\ge3$ with
$\|A_s\|<2\sqrt{2(\Delta-1)}$. For every prescribed bipartite graph, it gives
the sharp universal bound $\|A_s\|<2\sqrt{\Delta-1}$.  Analyzing the sequences of insertions and deletions gives an explicit tail
bound and fewer than $100n^3$ insertion attempts in expectation.

\item \textbf{Choosing rounding directions.}
Theorem~\ref{lem:core} uses the completion identity of
Lemma~\ref{lem:completion} to construct the response frame from local
two-dimensional inequalities. Algorithm~\ref{alg:main} tests both signs of
each direction in this frame to find a step that decreases the potential.

\item \textbf{Certified constant and bit-complexity proof.}
The constant rests on a fixed scalar inequality verified by an exact integer
certificate~\cite{certificate}.  Theorem~\ref{thm:complexity} and
Appendix~\ref{sec:precision} prove polynomial bit complexity for the rounding
procedure.  

\item \textbf{Partitions, paving, deterministic graph signings, and lifts.}
Corollary~\ref{cor:applications}, Propositions~\ref{prop:halving}
and~\ref{prop:paving}, and Corollary~\ref{cor:lifts} give spectral partitions,
small-diagonal paving, simultaneous adjacency/degree signings, and iterated
lift families from the deterministic higher-rank theorem.  Section
\ref{sec:repair} gives sharper adjacency-only randomized signings and
their lift consequences.
\end{itemize}

Sections~\ref{sec:preliminaries}--\ref{sec:applications} prove the rounding
theorem, establish its polynomial bit complexity, and give its applications.
Section~\ref{sec:repair} proves the randomized repair theorem by counting
the possible sequences of insertions and deletions. The appendices contain the scalar certificate and the
estimates needed to analyze deterministic rounding at finite precision.

\section{Mathematical preliminaries}\label{sec:preliminaries}
We work in the real inner-product space $\HH_d$ of Hermitian matrices, with
$\ip{U}{V}=\Tr(UV)$. The notation $U\preceq V$ means that $V-U$ is positive
semidefinite. Unsubscripted matrix norms are operator norms; $\norm{\cdot}_F$
denotes the Frobenius norm; vector norms are Euclidean. A rational Hermitian matrix has rational real and
imaginary parts. The real dimension of the input space is $p=d^2$, or
$p=d(d+1)/2$ for real symmetric inputs. A coordinate of $x\in[-1,1]^N$ is
\emph{active} when $|x_i|<1$. Let $M$ count the inputs retained after preprocessing and $m\le M$ the active coordinates.

Two trace inequalities let us extend the proof beyond rank one. The first
bounds the input map independently of dimension. The second replaces the
identities available for rank-one matrices.

\begin{lemma}[Trace contraction]\label{lem:trace}
If $M_i\succeq0$ and $\sum_i\Tr(M_i)M_i\preceq I$, then, for every Hermitian $Z$,
\begin{equation}\label{eq:tracecontraction}
 \sum_i|\Tr(M_iZ)|^2\le\norm Z_F^2.
\end{equation}
Consequently $\mathcal A v=\sum_iv_iM_i$ is a contraction from Euclidean to
Frobenius norm, and $\norm{M_i}\le1$.
\end{lemma}
\begin{proof}
Weighted Cauchy--Schwarz gives
$|\Tr(M_iZ)|^2\le\Tr(M_i)\Tr(M_iZ^2)$; summing proves the inequality.
Taking adjoints gives the bound on $\mathcal A$; $M_i^2\preceq\Tr(M_i)M_i\preceq I$ gives $\norm{M_i}\le1$.
\end{proof}

To state the second inequality, consider a trace-one positive matrix $H$.
Its \emph{purity} $\Tr(H^2)$ measures
how concentrated its spectrum is; it equals one at rank one. The overlap $\Tr(HL)$ is controlled
by this concentration even when both matrices have higher rank.

\begin{lemma}[Purity bounds]\label{lem:density}
If $H,L\succeq0$ have trace one, then $f=\Tr(HL)$ and $u=\Tr(H^2)$ satisfy
$0\le f\le1$ and $f^2\le u\le1$. In particular, for $a\ge0$ and $e\in\R$,
\begin{equation}\label{eq:support}
 \mathcal H(a,e):=\max_{0\le f\le1,\ f^2\le u\le1}(af-eu)
 =\begin{cases}a-e,&2e\le a,\\a^2/(4e),&2e>a.\end{cases}
\end{equation}
\end{lemma}
\begin{proof}
Trace Cauchy--Schwarz gives $f^2\le\Tr(H^2)\Tr(L^2)\le u$.
For $e\ge0$, maximize $af-ef^2$ over $[0,1]$; for $e<0$, take $f=u=1$.
\end{proof}

The trace inequalities will bound the potential's derivatives. To turn
these bounds into rounding steps, we use a linear map that records the
resulting first-order changes in the auxiliary matrices. We call it the \emph{response
map}. From this map we construct a finite family of directions, the response
frame. If a full-column-rank matrix $\Bcal$ maps directions to response
vectors, put
$\Gamma=\Bcal^*\Bcal$ and factor $\Gamma=L_\Gamma L_\Gamma^*$. The directions
\begin{equation}\label{eq:frame}
 h^{(j)}=\sqrt{m/Z}\,L_\Gamma^{-*}e_j,\qquad
 Z=\Tr(\Gamma^{-1}),\qquad
 \Sigma=\frac1m\sum_jh^{(j)}h^{(j)*}=\frac{\Gamma^{-1}}Z
\end{equation}
have $\Tr\Sigma=1$ and $\norm{h^{(j)}}\le\sqrt m$. The vectors $\Bcal h^{(j)}$
are mutually orthogonal and have equal norm. The inverse Gram matrix therefore records the covariance of directions whose
responses are isotropic on their subspace. Once we choose a suitable response map, a Cholesky factorization gives the
directions.

\section{The main proof}\label{sec:core}
Algorithm~\ref{alg:main} moves from the prescribed fractional point to a
vertex of the cube in two ways: it fixes a coordinate at a sign when possible
and otherwise takes a local step from the response frame. We first define
its potential and initialization, then
justify the endpoint test, the response frame, and the local search.
Section~\ref{sec:complexity} proves arithmetic and bit-complexity bounds for each step.

\subsection{Normalization and the potential}\label{sec:potential}
The preprocessing in line~\ref{line:preprocess} of Algorithm~\ref{alg:main}
reduces the number of active coordinates while preserving the signed sum
exactly. If the active inputs are linearly dependent,
move in a rational direction $h$
with $\sum_i h_iA_i=0$ until a coordinate reaches an endpoint. Repetition leaves
at most $M\le p$ active coordinates and a point $\bar x$ with
$\sum_i(\bar x_i-x_i^0)A_i=0$. Coordinates fixed in this way already have their
final signs. If $M=0$, the proof is complete. We keep this value of $M$ fixed
when choosing all subsequent tolerances.

After preprocessing, choose a rational scale $b$ slightly above $\sqrt V$ and put $M_i=A_i/b$.
Then $\sum_i\Tr(M_i)M_i\preceq I$, including for every retained subfamily.
Inputs with trace below a cutoff $\tau$ are rounded to their nearer endpoint
and set aside. During the subsequent rounding steps, coordinates within distance $\sigma$ of an
endpoint are also rounded there. We allow an increase of $\delta_E$ in an
endpoint test. The total error from these steps can be made arbitrarily small.
Appendix~\ref{sec:constantaccount} specifies the cutoffs and allowances;
Appendix~\ref{sec:precision} bounds their bit costs by a polynomial.
With $\bar x$ fixed, the retained inputs contribute discrepancy
$S(x)=\sum_i(x_i-\bar x_i)M_i$. 

To control this discrepancy, let $\psi$ be an even, strictly concave function on $(-1,1)$, with
$\psi(0)=1$, $\psi(\pm1)=0$, and $0\le\psi\le1$. We call this the \emph{weight function}: its value is the weight assigned to
a coordinate in the following constraints. For constants $c,\rho>0$, define the positive trace map
$\Eop_x(Z)=c\sum_i\psi(x_i)\Tr(M_iZ)M_i$ and the potential
\begin{equation}\label{eq:potential}
 R(x)=\min_{t\in\R,\ X,Y\succ0}
 \left\{t+\rho\Tr(X+Y):
 \begin{array}{l}
 X^{-1}+S(x)+\Eop_x(Y)\preceq tI,\\
 Y^{-1}-S(x)+\Eop_x(X)\preceq tI
 \end{array}\right\}.
\end{equation}
This adapts the regularized potential of Ezeunala and Jiang~\cite{EJ}.
For rank-one inputs, $\Tr(M_iZ)M_i=M_iZM_i$; the trace map extends their coupling terms to higher rank, with a different
weight function. The two inequalities bound the largest and smallest
eigenvalues of $S(x)$. Through $\Eop_x$, each coordinate weight enters both
inequalities. Rounding a coordinate sets its weight to zero and removes its
contribution to these coupling terms. The small trace penalty makes the minimizer unique and
smooth on each open face of the cube (Lemma~\ref{lem:kkt}).

The scalar in the constraints is a common upper bound for the two matrix
expressions; we call it the \emph{spectral level}. Positivity gives the lower
bound below. For the upper bound, set $X=Y=aI$ and minimize over $a$:
\begin{equation}\label{eq:initial}
 \norm{S(x)}\le R(x)\le\norm{S(x)}+2\sqrt{c+2d\rho}.
\end{equation}
Initially $S(\bar x)=0$. Thus the leading discrepancy coefficient is $2\sqrt c$;
the rounding procedure aims to preserve this initial bound. The error allowances
are chosen so that the initial value and all permitted increases together
remain below $4$. To obtain local descent, we choose a weight function of
the form $\psi(x)=\sqrt{1-x^2}P_b(x^2)$, with a fixed degree-six polynomial.
Its coefficients enter the local inequality in
Appendix~\ref{sec:certificate}. Integration by parts shows that normalization
fixes the weighted curvature: $\int_0^1(1-x)(-\psi''(x))\,dx=1$.
The choice of weight function therefore redistributes its negative second
derivative across the interval without changing the initial potential bound.

To implement the potential comparisons at finite precision, the pseudocode
uses certified bounds $L_x\le R(x)$ and $\overline U_y\ge R(y)$.
An \emph{upper witness} is a feasible choice of the auxiliary variables in
the potential; its objective value gives an upper bound on the potential.
The positive threshold $\widehat a(x)$ is a certified estimate of the local
curvature margin; its construction and the fixed working precision are given
in Appendix~\ref{sec:precision}. Frozen coordinates retain their assigned signs.

\begin{figure}[t]
\centering
\setlength{\fboxsep}{8pt}
\setlength{\fboxrule}{0.5pt}
\fcolorbox{darkblue!35}{black!1}{\begin{minipage}{\dimexpr\linewidth-2\fboxsep-2\fboxrule\relax}
\begin{algorithm}[Deterministic rounding]\label{alg:main}
\leavevmode\par\nobreak
\small
\begin{algorithmic}[1]
\Require Rational PSD matrices $A_1,\ldots,A_N$ and rational $x^0\in[-1,1]^N$.
\Ensure A signing satisfying Theorem~\ref{thm:main}.
\State Eliminate active linear dependencies while preserving the matrix sum;
       call the resulting state $\bar x$ and record endpoint signs.
       If no active coordinate remains, \Return the recorded signs.
       \label{line:preprocess}
\State Set $x\gets\bar x$, choose the parameters in~\eqref{eq:parameters},
       and normalize $M_i=A_i/b$. Freeze inputs with $\Tr(M_i)<\tau$
       at a nearest endpoint. Set $q\gets0$. \label{line:initialize}
\While{some coordinate is active}\label{line:outer}
  \State Freeze coordinates within $\sigma$ of an endpoint.
         If none remain, \textbf{break}. \label{line:boundary}
  \State Compute certified primal--dual data and $L_x$
         (Section~\ref{sec:currentstate}). \label{line:state}
  \If{an endpoint witness certifies an increase at most $\delta_E$}
         \label{line:endpoint}
    \State Take that endpoint and \textbf{continue} the outer loop,
           keeping $q$ (Section~\ref{sec:endpoint}).
  \EndIf
  \State Form the metric and response frame
         \eqref{eq:dictionary} and~\eqref{eq:frame}; cache the full optimizer
         derivatives and compute $\widehat a(x)$
         (Sections~\ref{sec:descent} and~\ref{sec:witness}). \label{line:frame}
  \Loop
    \State Set $s\gets2^{-q}$. \label{line:scale}
    \State Test both signs of every frame direction:
           form $y$ by rounding $x\pm sh^{(j)}$ to the fixed coordinate grid,
           reject candidates outside $|y_i|\le1-\sigma/2$,
           and compute valid upper witnesses $\overline U_y$
           (Lemma~\ref{lem:witness}). \label{line:trials}
    \If{some candidate satisfies
        $\overline U_y\le L_x-\widehat a(x)s^2/16$}
        \label{line:accept}
      \State Set $x\gets y$ and $q\gets\max\{0,q-1\}$;
             \textbf{break} to the outer loop. \label{line:restart}
    \EndIf
    \State Set $q\gets q+1$. \label{line:halve}
  \EndLoop
\EndWhile
\State \Return $x$, including the signs recorded during preprocessing.
       \label{line:return}
\end{algorithmic}
\end{algorithm}
\end{minipage}}
\end{figure}
All scans and tie breaks use a fixed index order. Lines~\ref{line:scale}--\ref{line:halve}
reuse the data at $x$; a rejected trial triggers no new optimization.
Section~\ref{sec:paired} proves acceptance and termination.

\subsection{Endpoint moves and compatible responses}\label{sec:endpoint}
We justify the endpoint test in line~\ref{line:endpoint}.
At the optimizer write $a_i=\Tr(M_iX)$ and $b_i=\Tr(M_iY)$.
If $c\psi(x_i)b_i\ge1-x_i$, changing $x_i$ to $1$ leaves the same primal
triple feasible: the first constraint changes by
$[(1-x_i)-c\psi(x_i)b_i]M_i\preceq0$, and the second decreases.
Likewise, rounding to $-1$ preserves feasibility if $c\psi(x_i)a_i\ge1+x_i$.
If neither test succeeds for any active coordinate, then
\begin{equation}\label{eq:light}
 c\psi(x_i)a_i<1+x_i,\qquad c\psi(x_i)b_i<1-x_i.
\end{equation}
We call this a \emph{light state}. The numerical endpoint test allows an increase of at most $\delta_E$.
This allowance ensures that a failed test implies the strict inequalities
above. Appendix~\ref{sec:linearsolves} gives the certified test. Thus line~\ref{line:frame} is reached only at a light state.

To choose a direction at a light state, we examine the minimizer. Both
constraints are tight. Their positive dual multipliers
$P,Q$ describe how the optimum weights the two spectral constraints. The
stationarity equations are
\begin{equation}\label{eq:dualcore}
 P=X\Eop_x(Q)X+\rho X^2,\qquad
 Q=Y\Eop_x(P)Y+\rho Y^2,\qquad \Tr(P+Q)=1.
\end{equation}
In particular the regularization keeps both multipliers positive definite.
The \emph{matrix responses} are the first-order changes in the auxiliary
matrices when the coordinates vary while $t$ stays fixed. Differentiating the
tight constraints gives the self-adjoint operator
\begin{equation}\label{eq:L}
 \Lop(U,V)=
 (X^{-1}UX^{-1}-\Eop_x(V),\ Y^{-1}VY^{-1}-\Eop_x(U)).
\end{equation}
Its inverse preserves positive semidefiniteness and has Frobenius operator
norm at most $\rho^{-1}$ (Lemma~\ref{lem:kkt}). Taking traces against the active
inputs reduces this matrix equation to two coefficients per coordinate.
After rescaling, the response equation becomes
\begin{equation}\label{eq:responsecore}
 (I-\mathsf T)y=J D_xv,\qquad J^*J=I,
 \qquad D_x=D_0^{1/2}\diag(t_i).
\end{equation}
Here $v\in\R^m$ is a coordinate direction, $y\in\R^{2m}$ records two
components of its matrix responses per coordinate, and $D_x$ is positive diagonal. The interaction
matrix $\mathsf T$ contains the trace Gramians of the inputs against $X$ and $Y$.
Appendix~\ref{sec:response} gives these matrices explicitly and proves that
$I-\mathsf T$ is invertible.

Equation~\eqref{eq:responsecore} couples responses at different coordinates,
so they cannot be chosen independently. We call responses \emph{compatible}
when they arise from a common coordinate direction through this equation.
A positive block-diagonal matrix $W$, called the \emph{response metric},
rescales them through its inverse. We apply~\eqref{eq:frame} to
\begin{equation}\label{eq:dictionary}
 \Bcal=W^{-1}(I-\mathsf T)^{-1}JD_x.
\end{equation}
The frame construction in line~\ref{line:frame} produces directions whose
weighted responses are orthogonal. We choose $W$ so that the average second variation over these compatible
responses gives descent.

\subsection{The completion identity and descent}\label{sec:descent}
The responses in line~\ref{line:frame} are coupled across coordinates,
whereas the inequalities we can verify are local. The following lemma uses
local $2\times2$ conditions to bound the response terms in the second
derivative, averaged over compatible responses.

\begin{lemma}[Two-slot completion]\label{lem:completion}
Let $\mathsf T,\mathsf E$ be real $2m\times2m$ matrices and let $J,\mathsf N$
have $m$ columns. Suppose $I-\mathsf T$ is invertible,
$(I-\mathsf T)y=J\xi$, $J^*J=I$, and
$\mathsf E\succeq\mathsf T^*\mathsf C\mathsf T$, where
$\mathsf C=\bigoplus_i C_iI_2$ and $C_i>0$.
Assume the columns $j_i,n_i$ of $J,\mathsf N$ are supported on the $i$th
coordinate pair; subscripts $ii$ denote the corresponding $2\times2$ blocks.
Choose $W_i\succ0$, real matrices $D_i$, and numbers $K_i>0,h_i\ge0$. Define
\begin{equation}\label{eq:completiondata}
 B_i=D_i^*W_i^{-1},\quad U_i=W_i^{-1}D_iW_i,\quad
 R_i=\frac{B_i^*B_i}{4C_i},\quad d_i^*=j_i^*B_i+n_i^*W_i.
\end{equation}
If for every $i$,
\begin{equation}\label{eq:completionconditions}
 \begin{gathered}
 \Tr(\mathsf E_{ii}W_i^2)+\Tr(D_i\mathsf T_{ii})\le h_i,\qquad
 \chi_i:=h_i+\Tr R_i>0,\\
 \Sym U_i-h_iI_2-\adj R_i-\frac{d_id_i^*}{4K_i}\succeq0,
 \end{gathered}
\end{equation}
there is a nonzero covariance of exact responses for which
\begin{equation}\label{eq:completionconclusion}
 \E(y^*\mathsf E y+\xi^*\mathsf N^*y)\le\sum_iK_i\E\xi_i^2.
\end{equation}
Here $\Sym U=(U+U^*)/2$ and $\adj R=(\Tr R)I_2-R$.
\end{lemma}

\paragraph{From compatible responses to the algorithm's frame.}
We next show how the response frame realizes the average bound in the
lemma. Normalize each block by
$\widehat W_i=\chi_i^{-1/2}W_i$ and $\widehat D_i=\chi_i^{-1}D_i$.
This divides $h_i,U_i,R_i$ by $\chi_i$ and leaves $K_i$ unchanged.
Write $W,D$ for the normalized block matrices, put $P_0=JJ^*$, and let
$\Pi$ project onto $\ker((I-P_0)(I-\mathsf T)W)$.
The two-square expansion in Appendix~\ref{sec:completion} proves the lemma
for $y=W\Pi g$ and $\xi=J^*(I-\mathsf T)W\Pi g$, where
$\E g=0$ and $\E gg^*=I$. They are compatible since $\ran((I-\mathsf T)W\Pi)\subseteq\ran J$.

For the response map $\Bcal$ in~\eqref{eq:dictionary},
$(I-\mathsf T)W\Bcal=JD_x$, so $\ran\Bcal=\ran\Pi$ and
\begin{equation}\label{eq:inducedcovariance}
 \Pi=\Bcal\Gamma^{-1}\Bcal^*,\qquad
 v=D_x^{-1}\xi=\Gamma^{-1}\Bcal^*g,\qquad
 \E vv^*=\Gamma^{-1}.
\end{equation}
Normalizing by $Z=\Tr(\Gamma^{-1})$ gives exactly the covariance $\Sigma$
of the finite frame~\eqref{eq:frame}. Every quadratic average in the lemma
therefore equals an average over the $m$ directions tested in
line~\ref{line:trials}; no random sampling is needed. With the weight
function and metric chosen below, the average second derivative is negative
when the gradient is small.

\begin{theorem}[Descent using the response frame]\label{lem:core}
Use the weight function and metric of Appendix~\ref{sec:certificate}. At every light state with $R(x)<4$ and $|x_i|\le1-\sigma$, the frame~\eqref{eq:frame} from~\eqref{eq:dictionary} satisfies $\Sigma\succeq\lambda_*I$ and
\begin{equation}\label{eq:curvaturecore}
 \norm{\nabla R(x)}_\infty\le g_*
 \quad\Longrightarrow\quad
 \frac12\Tr\bigl(\nabla^2R(x)\Sigma\bigr)\le-a_*.
\end{equation}
The positive constants $\lambda_*,g_*,a_*$ can be chosen with
inverse-polynomial dependence on $M,d$. The first three derivatives of $R$
are polynomially bounded on a fixed larger objective sublevel and the
slightly enlarged truncated face $|x_i|\le1-\sigma/2$.
\end{theorem}

\begin{proof}
We apply Lemma~\ref{lem:completion} to the second variation of the potential.
The remaining local inequality is uniform in the rank of each input.
Let $P,Q$ be the dual multipliers of~\eqref{eq:potential}. With the positive
scales $t_i,\ell_i,k_i$ defined in Appendix~\ref{sec:response}, the gradient takes the form
$\partial_iR=t_i(z_i-w_i)$, where
$z_i=\ell_i\Tr(M_iP)$ and $w_i=k_i\Tr(M_iQ)$.
A small gradient therefore means that these two weighted traces are nearly
balanced. To bound the second variation, hold $t$ fixed and solve the
two tight constraints along $x+sv$. Their objective majorizes $R(x+sv)$
and agrees with it at $s=0$. We differentiate the constraints twice and pair them with $P,Q$.
Stationarity then eliminates the second derivatives of $X,Y$
(Lemma~\ref{lem:variation}). At exact balance this gives
\begin{equation}\label{eq:balancedvariation}
 \tfrac12\nabla^2R[v,v]
 \le y^*\mathsf E y+\xi^*\mathsf N^*y-\sum_i\kappa_i\xi_i^2,
 \qquad \xi=D_xv.
\end{equation}
The first term is response energy, the second is the mixed response term,
and the negative term comes from the concavity of $\psi$.
Keeping the mixed term is essential to the constant. Lemmas~\ref{lem:purity}
and~\ref{lem:variation} establish this inequality and the energy bound
$\mathsf E\succeq\mathsf T^*\mathsf C\mathsf T$.

Lemma~\ref{lem:completion} bounds the first two terms by a local quadratic
expression in $\xi$. The algorithm uses the induced covariance~\eqref{eq:inducedcovariance}.
We choose the response metric so that this bound and the negative term from
the weight function have a strictly negative sum.

To check this bound, we examine the diagonal energy data. For rank-one
inputs these equal one. For arbitrary rank, they are overlaps and purities of trace-one PSD
matrices. They therefore lie in the region of Lemma~\ref{lem:density}. We bound the local response contribution by summing two values of
$\mathcal H$ in~\eqref{eq:support}.
This replaces the rank-one identities by an elementary quadratic maximization.
The explicit weight function and metric satisfy the resulting scalar inequalities with
a strict margin (Lemma~\ref{lem:scalar}). The margin absorbs the perturbation from exact balance. After this perturbation, the negative term from the weight function still
exceeds the response bound in magnitude by a fixed positive margin $\delta$. Consequently
\begin{equation}\label{eq:coregap}
 \tfrac12\Tr(\nabla^2R\,\Sigma)
 \le-\delta\sum_i(D_x)_{ii}^2\Sigma_{ii}
 \le-\delta\min_i(D_x)_{ii}^2.
\end{equation}
The last inequality uses $\Tr\Sigma=1$. Appendix~\ref{sec:robustness} gives the perturbation calculation.
Appendix~\ref{sec:precision} proves uniform lower bounds on the scales and on
$\Sigma$, and bounds the derivatives. These estimates give the stated
constants.
\end{proof}

\subsection{Paired moves and completion of the rounding}\label{sec:paired}
We now justify the paired trials in line~\ref{line:trials} and the acceptance
test in line~\ref{line:accept}. They exploit either a nonzero slope or the negative average
curvature from Theorem~\ref{lem:core}. If $\norm{\nabla R}_\infty>g_*$, then
$m^{-1}\sum_j|\nabla R\cdot h^{(j)}|^2\ge\lambda_*g_*^2$, so some frame column
has a definite slope. Its favorable sign decreases the potential.
If the gradient is small, average $R(x+sh^{(j)})$ and $R(x-sh^{(j)})$ over
the frame. The linear terms cancel, and~\eqref{eq:curvaturecore} gives
\[
 \frac1{2m}\sum_j\bigl(R(x+sh^{(j)})+R(x-sh^{(j)})\bigr)-R(x)
 \le-a_*s^2+O(M_3m^{3/2}s^3),
\]
where $M_3$ bounds the third derivative. Thus, at every sufficiently small scale, some tested step decreases $R$
by a positive multiple of $s^2$. A sufficient scale has an inverse-polynomial
lower bound.

The implementation tests these steps using feasible upper witnesses for
the trial values and a certified lower bound for the current value. One correction using the current linearization makes the witness error third
order in $s$ (Lemma~\ref{lem:witness}). The preceding argument therefore applies
without optimizing the potential at every trial. To compute the witnesses,
we differentiate the full optimality system and allow the spectral level to
vary. The responses that determine the frame hold this level fixed.

In certified arithmetic the test in line~\ref{line:accept} is
\begin{equation}\label{eq:acceptance}
 \overline U_y\le L_x-\widehat a(x)s^2/16,\qquad
 L_x\le R(x),\quad \overline U_y\ge R(y).
\end{equation}
Here $\widehat a(x)\in[a_x/2,a_x]$, with
$a_x=(\varpi/10)\min_i(D_0)_{ii}\ge a_*$ and $\varpi=1/50000$;
Appendix~\ref{sec:robustness} derives this local margin.
Thus every accepted candidate decreases the actual potential.
Appendix~\ref{sec:precision} chooses a uniform sufficient scale $s_c$ and
precision so that some candidate passes by scale $s_c/2$, including the
errors from computing the frame and rounding candidates to the coordinate grid.

Each endpoint operation in lines~\ref{line:initialize}, \ref{line:boundary},
and~\ref{line:endpoint} fixes a coordinate. Backtracking never needs a scale
smaller than an inverse polynomial, and the acceptance threshold makes each
accepted local move decrease $R$ by an inverse-polynomial amount. There are
therefore only polynomially many
such moves. The small-input, boundary, and endpoint allowances in these
three lines have combined cost at most
$M(\tau+\sigma+\delta_E)$. Together with~\eqref{eq:initial}, this proves
\begin{equation}\label{eq:roundingaccount}
 \Big\|\sum_i(s_i-x_i^0)A_i\Big\|
 \le b\left(2\sqrt{c+2d\rho}+M(\tau+\sigma+\delta_E)\right).
\end{equation}
The certified coefficient is $c=567/200$, so the leading coefficient is
$2\sqrt c=\sqrt{567/50}$. With the parameters in~\eqref{eq:parameters},
the bound is less than $\Cstar\sqrt V$, where $\Cstar=3.367912113<3.37$
(Theorem~\ref{thm:sharp}). These parameters also keep accepted states below the level required by the
descent theorem. The polynomial bit bound in Section~\ref{sec:complexity}
completes the proof of Theorem~\ref{thm:main}.

\section{Computational complexity}\label{sec:complexity}
To bound the cost of these moves, we analyze three computations in
Algorithm~\ref{alg:main}. We optimize at the
current state in line~\ref{line:state}, construct the frame and cache
derivatives in line~\ref{line:frame}, and certify trials in
line~\ref{line:trials}. We compute the current optimizer using matrix inversions, then reuse
factorizations in the linear solves that certify trials. We optimize again
only after the state changes.

\begin{theorem}[Arithmetic and bit complexity]\label{thm:complexity}
Let $p$ be the real dimension of the input matrix space. Exact preprocessing
uses $O(Np^2)$ field operations and leaves $M\le\min(N,p)$ active coordinates.
Let $L$ count accepted local moves in line~\ref{line:accept},
$T\le L+M+1$ count calls to line~\ref{line:state}, and $J_{\max}$ be the
largest exponent $q$ reached in line~\ref{line:scale}.
With classical dense arithmetic, the total cost is
\begin{equation}\label{eq:costintro}
 O\!\left(Np^2
 +T\,d(Md^2+d^3)\Lambda^2
 +(L+J_{\max})(Md^3+M^2d^2+M^3)\right),
\end{equation}
where $\Lambda=1+\log d+\log(1/\eta)$ and $\eta$ is the internal accuracy.
Parameters with $L=\operatorname{poly}(M,d)$ and $J_{\max}=O(\log(Md))$ give polynomial bit complexity in the input encoding.
\end{theorem}

The three terms in~\eqref{eq:costintro} account for preprocessing
(lines~\ref{line:preprocess}--\ref{line:initialize}), current-state solves
(line~\ref{line:state}), and cached frame and trial work
(lines~\ref{line:frame}--\ref{line:halve}), respectively.
The arithmetic model includes scalar square roots among
elementary operations; rational bracketing and certified comparisons are
accounted for in the bit-complexity statement. Norm scaling costs
$O(Nd^2+d^3\log(2d))$ operations and is absorbed by preprocessing. Preprocessing also limits the nonlinear phase to $M$ inputs. The retained
matrices are original inputs; every eliminated coordinate already has its
final sign. The streaming
updates and scale computation are detailed in Appendix~\ref{sec:preprocessing}.

\subsection{Computing the current state}\label{sec:currentstate}
Line~\ref{line:state} computes the optimizer and a certified value interval.
Fix $x$, abbreviate $\Eop_x$ to $\Eop$, and let $\beta$ be the infimum feasible
spectral level in~\eqref{eq:potential}, with the trace penalty omitted.
For $t>\beta$, start at $X_0=(tI-S)^{-1}$, $Y_0=(tI+S)^{-1}$ and iterate
\begin{equation}\label{eq:inverseiteration}
 X_{h+1}=(tI-S-\Eop(Y_h))^{-1},\qquad
 Y_{h+1}=(tI+S-\Eop(X_h))^{-1}.
\end{equation}
These monotone sequences converge to the least feasible pair $(X_t,Y_t)$.
The remaining optimization is one-dimensional:
$R(x)=\min_{t>\beta}F(t)$, where $F(t)=t+\rho\Tr(X_t+Y_t)$.
Differentiating the same recursion computes the dual data.

To estimate convergence, we represent the iterates as root blocks of the
resolvent of an auxiliary tree operator. Truncation at depth $h$ preserves
closed walks from the root through length $2h+1$. The finite and infinite
root resolvents therefore have matching moments through that order.
A Chebyshev approximation to the resolvent then gives an error exponent
proportional to $h\sqrt{(t-\beta)/t}$. The trace regularization keeps the minimizing level a polynomial distance
above $\beta$. Together, these estimates give depth $O(d\Lambda)$ and
$O(\Lambda)$ scalar probes.
Each layer uses $O(md^2+d^3)$ operations, giving
\begin{equation}\label{eq:datacost}
 C_D(m,d)=O\bigl(d(md^2+d^3)\Lambda^2\bigr).
\end{equation}
The lower endpoint of the certified interval~\eqref{eq:valueinterval}
supplies $L_x$ for line~\ref{line:accept}. Appendix~\ref{sec:oracle} bounds the gap and approximation error, and
justifies the guarded search. The tree and polynomial approximation appear
only in the analysis. The computation uses~\eqref{eq:inverseiteration} and
its derivative.

\subsection{Reusing the current linearization}\label{sec:witness}
Once the current optimizer is computed, lines~\ref{line:frame}
and~\ref{line:trials} share one factorization.
The response operator $\Lop$ in~\eqref{eq:L} can be solved through a
$2m\times2m$ coefficient system. Cache $XM_iX$, $YM_iY$, and their trace
Gramians, then factor this system once. Preparation costs
$C_F(m,d)=O(md^3+m^2d^2+m^3)$, and an additional right-hand side costs
$C_S(m,d)=O(md^2+m^2+d^3)$.
A solve with the Jacobian of the full optimality system reduces to two such
response solves and a scalar Schur complement
(Appendix~\ref{sec:linearsolves}). It therefore uses
the same factorization.

To use this factorization at trial points, let $\mathcal F(x,z)=0$ denote
the optimality system, with
$z=(t,X,Y,P,Q)$ and $J_z=\partial_z\mathcal F$ at the current state.
For each frame column $h$, cache the full optimizer derivative
$z_1=-J_z^{-1}\partial_x\mathcal F[h]$. At a trial point $y=x+sh$, predict
$z_p=z+sz_1$ and correct once with the current Jacobian:
\begin{equation}\label{eq:witness}
 z_c=z_p-J_z^{-1}\mathcal F(y,z_p),\qquad
 U(y)=t_c+\alpha+\rho\Tr(X_c+Y_c).
\end{equation}
Here $\alpha$ is the larger of the two Frobenius norms of the primal
constraint residuals. Certified arithmetic bounds this quantity to the
prescribed accuracy.
Adding this quantity to the spectral level makes the corrected triple
feasible, and hence $U(y)\ge R(y)$, whenever $X_c,Y_c\succ0$. Positive-definiteness tests reject invalid candidates. The
Frobenius bounds avoid an eigenvalue optimization at trial points.

\begin{lemma}[Third-order trial witness]\label{lem:witness}
For small $s$, uniformly on the truncated accepted-state region,
\begin{equation}\label{eq:witnessexpansion}
 U(x+sh)=R(x)+s\nabla R[h]+\tfrac12s^2\nabla^2R[h,h]
              +O(|s|^3\norm h^3).
\end{equation}
The error constant and the inverse of a valid neighborhood radius are
polynomially bounded in $M,d$.
\end{lemma}
\begin{proof}
The predictor cancels the first derivative of the optimality system. Its
residual is
$s^2D^2\mathcal F[(h,z_1),(h,z_1)]/2+O(|s|^3\norm h^3)$.
Implicit differentiation shows that the correction's quadratic term is
$s^2z_2/2$, where $z_2$ is the optimizer's second derivative.
Thus $z_c$ agrees with the true optimizer through second order, and its primal
residual is third order. The feasibility correction $\alpha$ is therefore also third order.
Since the objective is linear in $z$, the expansion follows. Uniformity follows
from the derivative and inverse-Jacobian bounds of Appendix~\ref{sec:precision}.
\end{proof}

Endpoint tests need no such correction: line~\ref{line:endpoint} uses the
current primal data directly. A small portion of the
endpoint allowance covers their numerical uncertainty, so every endpoint that preserves feasibility with the exact current data is
accepted. The complete tests are given in
Appendix~\ref{sec:linearsolves}.

\subsection{Total work and precision}
To obtain the total cost, we must count rejected scales as well as accepted
moves. The exponent $q$ persists through endpoint moves, decreases by at most one
after a local move (line~\ref{line:restart}), and increases by one after a
failed scale (line~\ref{line:halve}). The search therefore does not restart backtracking from the initial scale
after every move. If $k_\ell$ is the accepted exponent at local
search $\ell$ and $f_\ell$ is its number of failed scales, then
$k_\ell=k_{\ell-1}-u_\ell+f_\ell$, with $u_\ell\in\{0,1\}$.
Summing these recurrences cancels the intermediate exponents. Since the
active counts $m_\ell$ decrease, summation by parts also gives the weighted bound
\begin{equation}\label{eq:amortization}
 \sum_{\ell=1}^L m_\ell(1+f_\ell)
 \le2\sum_{\ell=1}^L m_\ell+MJ_{\max}.
\end{equation}
Caching all first responses costs $O(mC_S)=O(C_F)$ per local state; each
trial costs $O(C_S)$, including the matrix products and definiteness tests.
With~\eqref{eq:datacost}, this proves Theorem~\ref{thm:complexity}'s arithmetic bound.

For the bit-complexity bound, Appendix~\ref{sec:precision} gives uniform
lower bounds on the frame covariance, the allowable gradient imbalance, and
the decrease from curvature.
It uses these bounds to choose a minimum dyadic step and a rational working
precision for the acceptance tests. The resulting number of moves and bit
lengths are polynomially bounded.

Structured input can reduce individual costs. For rank-one matrices, write
$M_i=w_iw_i^*/a_i$ using a nonzero diagonal pivot
$a_i=(M_i)_{kk}$ and $w_i=M_ie_k$. Then
$\Tr(M_iXM_jX)=|w_i^*Xw_j|^2/(a_ia_j)$, and forming the response frame costs
$O(md^2+m^2d+m^3)$. The same improvement holds for a direct sum of a fixed number of rank-one
blocks. The remaining costs in~\eqref{eq:costintro} still apply.

\section{Partitions and graph signings}\label{sec:applications}
The applications use the fact that each PSD input receives one sign.
A vector may represent a frame element, while a higher-rank matrix may
encode several constraints that must receive the same sign.
We use the sharper coefficient $\Cstar$ of Theorem~\ref{thm:sharp} throughout.

\begin{corollary}[Partitions and signings]\label{cor:applications}
These constructions take deterministic polynomial time.
\begin{enumerate}[label=(\roman*)]
\item If rational PSD inputs satisfy $\sum_iA_i=T\succeq0$ and $\Tr(A_i)\le\varepsilon$, there is a partition $S_+\sqcup S_-$ with
\[
 \max_{\omega\in\{+,-\}}
 \Big\|\sum_{i\in S_\omega}A_i-\tfrac12T\Big\|
 <\tfrac{\Cstar}{2}\sqrt{\varepsilon\norm T}
 <1.69\sqrt{\varepsilon\norm T}
\]
when $T\ne0$. For $A_i=v_iv_i^*$ and $T=I$, this is the Kadison--Singer partition problem. 
\item Every finite simple graph of maximum degree $\Delta>0$ has an edge signing whose signed adjacency and signed-degree matrices satisfy
\[
 \max\{\norm{A_s},\norm{D_s}\}<2\sqrt2\,\Cstar\sqrt\Delta<9.53\sqrt\Delta.
\]
For bipartite graphs the bound is $2\Cstar\sqrt\Delta<6.74\sqrt\Delta$.
The graph bounds also hold for rational edge weights in $[0,1]$, with
positive maximum weighted degree replacing $\Delta$, and for differences
from a prescribed rational fractional signing. The zero-weight case has zero
discrepancy.
\end{enumerate}
\end{corollary}

\subsection{Spectral partitions}
If $\sum_iA_i=T$ and $\Tr(A_i)\le\varepsilon$, positivity gives
$\sum_i\Tr(A_i)A_i\preceq\varepsilon T$.
Starting at $x^0=0$, the signed sum is twice the deviation of either part
from $T/2$, so
Theorem~\ref{thm:sharp} proves the partition assertion in
Corollary~\ref{cor:applications}. The same argument applies to every subfamily $S$, whose total matrix is
$T_S=\sum_{i\in S}A_i$. We can therefore apply it repeatedly to halve the
parts.

\begin{proposition}[Multi-way partitions]\label{prop:halving}
Let $A_i\succeq0$ be rational, $\Tr(A_i)\le\varepsilon$, and $t_0=\norm{\sum_iA_i}>0$. If $2^k\varepsilon\le t_0/67$, a deterministic polynomial-time algorithm partitions the inputs into $2^k$ parts, each satisfying
\begin{equation}\label{eq:halving}
 \Big\|T_S-2^{-k}\sum_iA_i\Big\|
 \le11.5\sqrt{\varepsilon t_0}\,2^{-k/2}
\end{equation}
\end{proposition}
\begin{proof}
Apply the two-part theorem recursively. If $e_j$ bounds the deviation of a
depth-$j$ part from $2^{-j}\sum_iA_i$, its norm is at most
$2^{-j}t_0+e_j$. Since $\Cstar/2<1.684$, the next split obeys
\begin{equation}\label{eq:halvingrecurrence}
 e_0=0,\qquad
 e_{j+1}=\tfrac12e_j+1.684\sqrt{\varepsilon(2^{-j}t_0+e_j)}.
\end{equation}
Appendix~\ref{sec:halvingproof} verifies that this recurrence gives
\eqref{eq:halving} under the stated restriction on $k$.
\end{proof}

For a positive contraction, spectral partitioning also gives a paving by
coordinate sets. We write $P_S$ for coordinate projection onto $S$.

\begin{proposition}[Paving positive contractions]\label{prop:paving}
Let $\gamma,\epsilon>0$ be rational. Suppose $T\succeq0$, $\norm T\le1$, and $T_{ii}\le\gamma$.
If $T=C^*C$ is supplied with rational $C$ and $268\gamma\le\epsilon\le1$, one can find a partition into $r\le4/\epsilon$ parts with $\norm{P_STP_S}\le\epsilon$.
For a rational $T$ supplied without a factorization, $269\gamma\le\epsilon\le1$ suffices, with $r<4.02/\epsilon$. Both algorithms have polynomial bit complexity in their rational input encoding.
\end{proposition}
\begin{proof}
For factored input, partition the column outer products of $C$; their partial
sums have the same nonzero spectra as the corresponding coordinate
compressions of $T$. A rational square-root approximation gives the unfactored
case. Appendix~\ref{sec:pavingproof} gives the quantitative estimates and the
finite-precision factorization argument.
\end{proof}

\subsection{Graph signings}
For graphs, we encode adjacency and degree together so that both receive the
same edge signs. Let $A,D$ be a simple graph's unsigned adjacency and degree matrices. For each edge $e=uv$, set
\begin{equation}\label{eq:graphencoding}
 B_e^\pm=(e_u\pm e_v)(e_u\pm e_v)^*,\qquad
 \widehat B_e=B_e^+\oplus B_e^-.
\end{equation}
Each $\widehat B_e$ is rational PSD of rank two and trace four. For the signed-degree matrix $D_s$,
\begin{equation}\label{eq:graphvariance}
 \sum_es_e\widehat B_e=(D_s+A_s)\oplus(D_s-A_s),\qquad
 \sum_e\Tr(\widehat B_e)\widehat B_e=4[(D+A)\oplus(D-A)].
\end{equation}
Both $D\pm A$ are positive semidefinite and have norm at most $2\Delta$. Theorem~\ref{thm:sharp} therefore gives a direct-sum norm below $\Cstar\sqrt{8\Delta}$. Half the difference of the two blocks gives the adjacency matrix; half their
sum gives the degree matrix. Thus
$\norm{A_s},\norm{D_s}<2\sqrt2\,\Cstar\sqrt\Delta$.
This is the general graph bound of Corollary~\ref{cor:applications}. For rational weights $\omega_e\in[0,1]$, use $\omega_e\widehat B_e$ and $\omega_e^2\le\omega_e$; the trace scale is at most $8\Delta_\omega$. Fractional starts give the same statements for differences from $A_{x^0},D_{x^0}$.

For a bipartite graph with classes $U,W$, use instead $v_e=e_u\oplus e_v$ in $\R^U\oplus\R^W$, and input $v_ev_e^*$. The unsigned sum has diagonal blocks bounded by $\Delta I$. For any PSD block matrix,
$\left(\begin{smallmatrix}X&Y\\Y^*&Z\end{smallmatrix}\right)\preceq2\diag(X,Z)$,
because subtracting the block matrix from the right-hand side gives its
conjugate by $\diag(I,-I)$. Since each input has trace two, the trace scale is at most $4\Delta$. The off-diagonal signed block has norm below $2\Cstar\sqrt\Delta$, and the bipartite signed adjacency has exactly that norm. The diagonal blocks give the same signed-degree bound. Weights and fractional starts work as above.

\paragraph{Comparison with the oriented encoding.}\phantomsection\label{sec:oriented}
An oriented encoding also controls signed in-degrees and out-degrees. Pair the odd-degree vertices and add the pairing edges. Orient an Euler tour,
then delete the added edges. The in-degrees and out-degrees are both at most $\lceil\Delta/2\rceil$. Use $e_u\oplus e_v$ for each oriented edge $u\to v$ in two copies of the vertex space. The PSD block bound gives trace scale at most $4\lceil\Delta/2\rceil$, hence a signing with
\[
 \norm{A_s}<4\Cstar\sqrt{\lceil\Delta/2\rceil},\qquad
 \max_v\{|d_s^{\rm out}(v)|,|d_s^{\rm in}(v)|\}
 <2\Cstar\sqrt{\lceil\Delta/2\rceil}.
\]
For weighted graphs any orientation gives the analogous bounds with
$\Delta_\omega$ inside the square roots. In the unweighted case, the direct-sum
encoding~\eqref{eq:graphencoding} gives an adjacency bound no larger than the
oriented bound, and strictly smaller when $\Delta$ is odd. The oriented version also controls signed in-degrees and out-degrees
separately. Their sum bounds the signed degrees, so both encodings control
adjacency and degree simultaneously.

For the adjacency matrix alone, the randomized algorithm of
Section~\ref{sec:repair} improves the general-graph coefficient, relative to
$\sqrt\Delta$, from $2\sqrt2\,\Cstar<9.53$ to
$2\sqrt2\sqrt{(\Delta-1)/\Delta}<2.83$. For bipartite graphs it improves
$2\Cstar<6.74$ to $2\sqrt{(\Delta-1)/\Delta}<2$. It does not control signed
degrees and does not accept weights or fractional starts.

\subsection{Iterated lifts}\label{sec:lifts}
Graph signings also control the new eigenvalues of graph lifts. Given a
signing, make two copies of the graph's vertex set. Keep the two
copies of each positive edge within their respective sheets. Cross the two
copies of each negative edge between the sheets. The adjacency of this $2$-lift is
$\left(\begin{smallmatrix}A_+&A_-\\A_-&A_+\end{smallmatrix}\right)$.
Under the orthogonal change of basis $(u,v)\mapsto((u+v)/\sqrt2,(u-v)/\sqrt2)$,
this adjacency matrix becomes $A\oplus A_s$. Therefore repeated applications of the graph signing algorithm give the following.

\begin{corollary}[Deterministic lift families]\label{cor:lifts}
From a $d$-regular graph $G_0$ on $n_0$ vertices, one can construct a sequence of $2$-lifts $G_k$ on $n_02^k$ vertices, in time polynomial in the output size, such that
\[
 \lambda(G_k)\le\max\{\lambda(G_0),\,2\sqrt2\,\Cstar\sqrt d\}.
\]
For bipartite $G_0$, replace the second term by $2\Cstar\sqrt d$. Here $\lambda$ excludes one copy of $d$ and, for bipartite graphs, one copy of $-d$, and takes the maximum absolute value of the remaining eigenvalues (zero if none remain). If $G_0$ is connected, so are all lifts, with a uniform spectral gap for $d\ge91$ (bipartite: $d\ge46$).
\end{corollary}
\begin{proof}
The spectral decomposition above proves the bound inductively. The bound on
new eigenvalues is below $d$ in the stated ranges because $8\Cstar^2<91$ and $4\Cstar^2<46$. No new eigenvalue $d$ can then appear, so connectedness is preserved. The sum of polynomial costs over the successive sizes is polynomial in $n_02^k$.
\end{proof}

Corollary~\ref{cor:rep-lifts} lowers the degree thresholds to $d\ge7$ and,
for bipartite graphs, $d\ge3$, using the randomized algorithm of
Section~\ref{sec:repair}.

\section{Randomized repair at the Bilu--Linial scale}\label{sec:repair}

The graph bounds in the preceding section follow from the general
matrix-discrepancy theorem, Theorem~\ref{thm:main}. When we need to control
only the signed adjacency matrix, we can use the graph structure to obtain
smaller constants.  We give a Las Vegas
algorithm which, for every graph of maximum degree at most $\Delta$, finds a
signing with
\[
        \lVert A_s\rVert<2\sqrt{2(\Delta-1)},
\]
and, for every bipartite graph, reaches the sharp universal threshold
\[
        \lVert A_s\rVert<2\sqrt{\Delta-1}.
\]
The general result is within a factor $\sqrt2$ of the Bilu--Linial conjecture;
the bipartite result reaches the conjectured bound. These smaller constants
come with a narrower scope. The algorithm is randomized and controls only the
adjacency matrix. It does not allow weights or fractional starts, or give
simultaneous bounds on signed degrees.

The algorithm inserts vertices one at a time while maintaining positive
definite matrices that encode the spectral constraints. A proposed insertion
must satisfy these constraints; when it does, a ratio of determinants gives
its acceptance probability. After a rejection, the algorithm chooses a
neighbor for deletion using diagonal entries of the inverse constraint
matrices. It then restores that neighbor recursively and retries the original
insertion. We call a finite record of these insertions and deletions a
\emph{repair history}. Determinant identities bound the probability of each
history. Counting trees of bounded depth then shows that the probability of
long histories decays exponentially. The bound for a single deletion gives
both spectral bounds and identifies the factor of two that causes the
$\sqrt2$ gap in the general case.

Throughout this section, $G=(V,E)$ is a finite simple graph with
$|V|=n\ge 1$. The integer $\Delta\ge 3$ bounds its maximum degree. Taking
$\Delta=3$ also covers graphs of maximum degree at most two. For $K\subseteq V$ and a
signing $s$ of the edges of $G[K]$, write $A_K=A_K(s)$ for the signed adjacency
matrix of $G[K]$. Empty determinants equal one. For a radius $r>0$ put
\begin{equation}\label{eq:rep-Q}
Q_K^{\pm}=I\pm A_K/r,\qquad N_K=Q_K^{+}Q_K^{-}=I-A_K^{2}/r^{2}.
\end{equation}
The algorithm runs in one of two modes, indexed by $\nu\in\{1,2\}$.
\begin{itemize}
\item In the \emph{one-sided mode} ($\nu=1$), $W_K=\det Q_K^{-}$, and the
state is \emph{good} when $Q_K^{-}\succ 0$, that is, when
$\lambda_{\max}(A_K)<r$.
\item In the \emph{two-sided mode} ($\nu=2$), $W_K=\det N_K
=\det Q_K^{+}\det Q_K^{-}$, and the state is good when $N_K\succ 0$. Since
$Q_K^{\pm}$ commute and $N_K$ has eigenvalues $(1-\lambda/r)(1+\lambda/r)$
for the eigenvalues $\lambda$ of $A_K$, this holds exactly when both
$Q_K^{\pm}\succ 0$, that is, when $\lVert A_K\rVert<r$.
\end{itemize}
A good state admits linearly independent unit vectors whose Gram matrix is
$Q_K^{-}$ (and, in the two-sided mode, a second such family for $Q_K^{+}$).
Adjacent vertices then have inner products $\mp s_{uw}/r$ and distinct
nonadjacent vertices are orthogonal. The determinants of these Gram matrices are the squared volumes of the
corresponding families of vectors. This interpretation motivates the
acceptance and deletion rules below; the proofs use the determinants
directly.

\begin{theorem}[Main graph theorem: randomized repair]\label{thm:repair}
Let $r_*=2\sqrt{\Delta-1}$ in the one-sided mode and
$r_*=2\sqrt{2(\Delta-1)}$ in the two-sided mode. Algorithm~\ref{alg:repair},
run with $r=r_*$, terminates almost surely and returns a signing $s$ of $E$
with
\begin{equation}\label{eq:rep-bounds}
\lambda_{\max}(A_s)<2\sqrt{\Delta-1}\quad(\nu=1),\qquad
\lVert A_s\rVert<2\sqrt{2(\Delta-1)}\quad(\nu=2).
\end{equation}
If $G$ is bipartite, the one-sided output also satisfies
$\lVert A_s\rVert<2\sqrt{\Delta-1}$. In either mode, the number $\mathsf T$ of
insertion attempts satisfies, for every integer $t\ge n$,
\begin{equation}\label{eq:rep-tail}
\Pr(\mathsf T>t)\le \min\left\{1,\;5^{n}\Bigl(1+\frac{1}{16n^{2}}\Bigr)^{-(t-n)/2}\right\},
\qquad \mathbb E\,\mathsf T<100\,n^{3}.
\end{equation}
If instead $r=(1+\epsilon)r_*$ with $0<\epsilon\le 1$, then
$\Pr(\mathsf T>t)\le \min\{1,\;4^{n}(1+\epsilon)^{-(t-n)}\}$ and
$\mathbb E\,\mathsf T=O(n/\epsilon)$. Each attempt can be carried out in exact
arithmetic with polynomial expected bit complexity: over $\mathbb Q$ in the
two-sided mode and for bipartite graphs in the one-sided mode, and over
$\mathbb Q(\sqrt{\Delta-1})$ in the one-sided mode in general.
\end{theorem}

The two-sided bound improves the coefficient $2\sqrt 2\,C_*<9.53$ of
Corollary~\ref{cor:applications}(ii) to
$2\sqrt2\,\sqrt{(\Delta-1)/\Delta}<2.83$ for the adjacency matrix alone.
Bilu and Linial~\cite{BL} conjectured that every $\Delta$-regular graph
has a signing with $\lVert A_s\rVert\le 2\sqrt{\Delta-1}$; our two-sided mode
misses this by the factor $\sqrt2$.  Lemma~\ref{lem:rep-charge} explains this loss. When the determinant weight constrains both ends of the spectrum, the
deletion bound contains an extra factor of two.

\paragraph{Relation to MSS and Cohen.}
Marcus, Spielman, and Srivastava~\cite{MSS1} proved the stronger one-sided
existence statement
\[
        \lambda_{\max}(A_s)\le \rho(\mu_G),
\]
where $\rho(\mu_G)$ is the largest root of the matching polynomial.  The
Heilmann--Lieb bound gives $\rho(\mu_G)<2\sqrt{\Delta-1}$, and for bipartite
$G$ spectral symmetry turns this into the corresponding two-sided bound.
Our bipartite theorem addresses the associated \emph{prescribed-graph signing
problem}: the graph $G$ is part of the input, and the algorithm chooses signs
on the existing edge set $E(G)$.

Cohen's deterministic construction~\cite{Cohen} solves a different input
problem. Given a size and degree, it uses an interlacing family to select
complete bipartite matchings whose union is a Ramanujan multigraph. It does
not sign the existing edges of an arbitrary prescribed bipartite graph.
Cohen also distinguishes this construction from the signing-based
interlacing family of~\cite{MSS1}. The bipartite part of
Theorem~\ref{thm:repair} instead takes the graph as input and signs its edges.
The algorithm is randomized, but its decisions are exact and, in the
bipartite case, rational. It does not use interlacing polynomials or evaluate
expected characteristic polynomials. For a prescribed $d$-regular bipartite
Ramanujan graph, its output specifies a Ramanujan $2$-lift of that base.
This comparison concerns the input problems and proof methods.

The following proposition records that the universal constant in the
one-sided, and hence bipartite, statement cannot be improved.

\begin{proposition}[Sharpness of the universal one-sided bound]\label{prop:rep-sharp}
Let $\Delta\ge3$ and $c<2\sqrt{\Delta-1}$. There is a connected bipartite
$\Delta$-regular graph none of whose signings satisfies $\lVert A_s\rVert\le c$,
and a connected $\Delta$-regular graph none of whose signings satisfies
$\lambda_{\max}(A_s)\le c$.
\end{proposition}

\begin{proof}
We may assume $c\ge0$. Suppose that every connected bipartite $\Delta$-regular
graph has a signing with $\lVert A_s\rVert\le c$. Let $G_0=K_{\Delta,\Delta}$,
and let $G_{k+1}$ be the $2$-lift of $G_k$ defined by such a signing, as in
Section~\ref{sec:lifts}. The spectrum of $G_{k+1}$ is that of $G_k$ together
with that of $A_s$, and a lift of a bipartite graph is bipartite. Since the
eigenvalues of $K_{\Delta,\Delta}$ other than $\pm\Delta$ vanish, every
eigenvalue of $G_k$ other than one copy each of $\pm\Delta$ lies in $[-c,c]$.
As $c<\Delta$, the eigenvalue $\Delta$ is simple, so $G_k$ is connected and the
induction may continue. The graph $G_k$ has $\Delta2^{k+1}$ vertices and degree
$\Delta$, so its diameter tends to infinity with $k$. By the Alon--Boppana
bound in the form of Nilli~\cite{Nilli}, its second largest eigenvalue is at
least $2\sqrt{\Delta-1}-o(1)$, contradicting the bound $c$.

For the one-sided assertion, start instead from $G_0=K_{\Delta+1}$, whose
eigenvalues other than $\Delta$ equal $-1$, and lift by signings with
$\lambda_{\max}(A_s)\le c$. Every eigenvalue of $G_k$ other than one copy of
$\Delta$ is then at most $c<\Delta$, so $G_k$ is connected, and the same bound
gives the contradiction.
\end{proof}

This sharpness is asymptotic: every finite graph satisfies the strict
inequalities of Theorem~\ref{thm:repair}. Some graphs admit much smaller
signings: the bound $\rho(\mu_G)$ of~\cite{MSS1}, for example, can lie well
below $2\sqrt{\Delta-1}$. Whether the repair mechanism can reach this
graph-specific bound is discussed in Section~\ref{sec:discussion}.

\subsection{Gram extensions and leave-one-out residuals}

To prove the algorithmic bound, we begin with two Schur-complement
identities. The first gives the
weight of an insertion; the second gives the weight of a deletion.

\begin{lemma}[Extension and deletion]\label{lem:rep-schur}
Let $K$ be good, let $v\notin K$, and let $c\in\mathbb R^{K}$ have
$c_u=s_{uv}/r$ for $u\in N(v)\cap K$ and $c_u=0$ otherwise. Put
$q_{\pm}=c^{*}(Q_K^{\pm})^{-1}c$.
\begin{enumerate}
\item[(i)] $\det Q_{K+v}^{\pm}=\det Q_K^{\pm}\,(1-q_{\pm})$. Consequently
$K+v$ is good exactly when $q_-<1$ (and, in the two-sided mode, also
$q_+<1$), and then $W_{K+v}=\alpha\,W_K$, where $\alpha=1-q_-$ in the
one-sided mode and $\alpha=(1-q_-)(1-q_+)$ in the two-sided mode.
\item[(ii)] For $u\in K$ put
$\delta_u^{\pm}=1/\bigl((Q_K^{\pm})^{-1}\bigr)_{uu}$. Then
$0<\delta_u^{\pm}\le 1$ and $\det Q_K^{\pm}=\det Q_{K-u}^{\pm}\,\delta_u^{\pm}$.
\end{enumerate}
\end{lemma}

\begin{proof}
Order $v$ last. The new column of $Q_{K+v}^{\pm}$ is $\pm c$ with unit diagonal
entry, so its Schur complement is $1-c^{*}(Q_K^{\pm})^{-1}c=1-q_{\pm}$. This proves the determinant identity. A bordered matrix with a positive
definite leading block is positive definite exactly when its Schur complement
is positive. For (ii), order $u$ last and write
$Q_K^{\pm}=\bigl(\begin{smallmatrix}B&a\\a^{*}&1\end{smallmatrix}\bigr)$. Block
elimination gives $\det Q_K^{\pm}=\det B\,(1-a^{*}B^{-1}a)$ and
$((Q_K^{\pm})^{-1})_{uu}=(1-a^{*}B^{-1}a)^{-1}$. Since $B\succ0$, the Schur
complement lies in $(0,1]$. Geometrically, $\delta_u^{\pm}$ is the squared
distance from the unit vector of $u$ to the span of the others. We call this
squared distance the \emph{leave-one-out residual}.
\end{proof}

\subsection{The algorithm}

We use the insertion determinant ratio for acceptance and the inverse
diagonal entries for deletion.

\begin{center}
\fbox{\begin{minipage}{0.94\linewidth}
\refstepcounter{theorem}\label{alg:repair}\textbf{Algorithm~\thetheorem} (Randomized repair).\\[2pt]
\textit{Input:} a graph $G$, a vertex order $(v_1,\dots,v_n)$, a mode
$\nu\in\{1,2\}$, and a radius $r$.\\
\textit{Output:} a signing $s$ of $E$ with $G$ good at radius $r$.\\[2pt]
Start with $K=\emptyset$ and no signs. For $i=1,\dots,n$, call
$\textsc{Fill}(v_i)$.\\[2pt]
$\textsc{Fill}(v)$: repeat the following.
\begin{enumerate}
\item Draw independent uniform signs $s_{uv}$ for $u\in N(v)\cap K$. If $K+v$
is good, set $\alpha=W_{K+v}/W_K$; otherwise set $\alpha=0$.
\item With probability $\alpha$, add $v$ to $K$ with these signs and return.
\item Otherwise discard the proposed signs. Choose $u\in N(v)\cap K$ with
probability $\lambda_u/\sum_{w\in N(v)\cap K}\lambda_w$, where
\[
\lambda_u=\bigl((Q_K^{-})^{-1}\bigr)_{uu}\ \ (\nu=1),\qquad
\lambda_u=\bigl((Q_K^{-})^{-1}\bigr)_{uu}+\bigl((Q_K^{+})^{-1}\bigr)_{uu}
=2\,(N_K^{-1})_{uu}\ \ (\nu=2).
\]
Remove $u$ from $K$ together with all its signs, call $\textsc{Fill}(u)$,
and retry $v$.
\end{enumerate}
\end{minipage}}
\end{center}

If $v$ has no active neighbour, then $c=0$, $\alpha=1$, and the insertion
succeeds; a rejection therefore always has a neighbour to delete. The identity
$(Q^{-})^{-1}+(Q^{+})^{-1}=2(I-A^{2}/r^{2})^{-1}$ in step~3 holds because the
two factors commute. The deletion choice uses fresh randomness and, given the
current state, is independent of the discarded proposal.

At every finite time, the active set is disjoint from the labels on the
recursion stack. Each child is chosen from the active vertices and removed
before its call begins. Stack labels are therefore distinct, and the depth
is at most $n$. Induction on the finite call tree shows what any returning
call leaves active: its own vertex and exactly those vertices that were
active at its entry. The signs may differ. This argument does not assume that
other calls terminate. Every outer call therefore begins at an inactive
vertex. A completed execution makes exactly $n$ outer calls and ends with
$K=V$ good. In the one-sided mode this gives $\lambda_{\max}(A_s)<r$; in the
two-sided mode it gives $\lVert A_s\rVert<r$. For bipartite $G$, consider the diagonal matrix equal to $+1$ and $-1$ on
the two classes. Conjugation by this matrix maps $A_s$ to $-A_s$. The signed adjacency spectrum is therefore symmetric, and
$\lVert A_s\rVert=\lambda_{\max}(A_s)$.
This proves~\eqref{eq:rep-bounds} for every completed execution; it remains to
bound the number of attempts.

\subsection{The deletion bound}

To bound the number of attempts, we first analyze a single deletion.
The next lemma compares a rejection-and-deletion probability with the
determinant weights before and after the deletion. Write $\pi=\nu/r^{2}$
for the \emph{deletion factor} in this comparison. At the thresholds of
Theorem~\ref{thm:repair},
$\pi=1/(4(\Delta-1))$ in both modes.

\begin{lemma}[Deletion bound]\label{lem:rep-charge}
Let $K$ be good and let $v\notin K$ be the vertex being filled. For every
$u\in N(v)\cap K$,
\begin{equation}\label{eq:rep-charge}
W_K\,\Pr(\text{the next attempt rejects and deletes }u\mid K,s)
\le \frac{\lambda_u W_K}{r^{2}}\le \pi\,W_{K-u}.
\end{equation}
\end{lemma}

\begin{proof}
The rejection probability for a given proposal is $1-\alpha$. In the
one-sided mode $1-\alpha=\min\{q_-,1\}\le q_-$. In the two-sided mode, $1-(1-q_-)_+(1-q_+)_+\le q_-+q_+$ for
$q_\pm\ge0$. If both are below one, the left side is $q_-+q_+-q_-q_+$.
Otherwise it equals one. Independent
unbiased signs give $\mathbb E\,cc^{*}=r^{-2}$ times the coordinate projection
onto $N(v)\cap K$, so
\[
\mathbb E\,q_{\pm}=\frac{1}{r^{2}}\sum_{w\in N(v)\cap K}
\bigl((Q_K^{\pm})^{-1}\bigr)_{ww},\qquad
\mathbb E(1-\alpha)\le\frac{1}{r^{2}}\sum_{w\in N(v)\cap K}\lambda_w .
\]
Multiplying by the independent deletion probability
$\lambda_u/\sum_w\lambda_w$ proves the first inequality. By
Lemma~\ref{lem:rep-schur}(ii), $W_K\lambda_u=W_{K-u}$ in the one-sided mode.
In the two-sided mode,
\[
W_K\lambda_u=\det Q^{+}_K\det Q^{-}_{K-u}+\det Q^{-}_K\det Q^{+}_{K-u}
=W_{K-u}\,(\delta_u^{+}+\delta_u^{-})\le 2\,W_{K-u}.
\qedhere
\]
\end{proof}

The factor two in the last display is the only loss relative to the
Bilu--Linial conjecture: with a deletion factor of $1/r^{2}$ in the two-sided
mode, the counting below would give $r_*=2\sqrt{\Delta-1}$.

\subsection{Transition bounds and repair histories}

We now combine these one-step bounds along a repair history. For a
signing $s$ of $G[K]$ define
\begin{equation}\label{eq:rep-weight}
\Phi_K(s)=2^{-|E(G[K])|}\,W_K(s)\,\mathbf 1\{K\text{ is good under }s\}.
\end{equation}
A good matrix $Q_K^{-}$ or $N_K$ is positive definite with diagonal entries at
most one, so Hadamard's inequality gives $W_K\le1$ and hence
$\sum_s\Phi_K(s)\le1$.

\begin{lemma}[Transition domination]\label{lem:rep-kernel}
Suppose the subprobability mass $\mu$ of arriving at active set $K$, with $v$
being filled, satisfies $\mu(s)\le C\,\Phi_K(s)$ for every signing $s$ of
$G[K]$. Then the mass of a successful insertion of $v$ ending at a specified
signing of $G[K+v]$ is at most $C\,\Phi_{K+v}$ of that signing, and the mass
of a rejection deleting a specified $u$, ending at a specified signing $\varrho$
of $G[K-u]$, is at most $C\pi\,\Phi_{K-u}(\varrho)$.
\end{lemma}

\begin{proof}
Let $a=|N(v)\cap K|$. A specified successful extension $(s,s')$ has mass at
most $C\,2^{-|E(G[K])|}W_K(s)\,2^{-a}\,\alpha$, which equals
$C\,\Phi_{K+v}(s,s')$ by Lemma~\ref{lem:rep-schur}(i); rejected proposals
carry zero weight on both sides. For a deletion, let $b=\deg_{G[K]}(u)$. By
Lemma~\ref{lem:rep-charge}, each good $s$ restricting to $\varrho$ contributes
at most $C\pi\,2^{-|E(G[K])|}W_{K-u}(\varrho)$, and there are at most $2^{b}$
such $s$. Their total is $C\pi\,2^{b-|E(G[K])|}W_{K-u}(\varrho)
=C\pi\,\Phi_{K-u}(\varrho)$. If $\varrho$ is bad, no good $s$ restricts to it,
since principal submatrices of positive definite matrices are positive
definite; both sides vanish. The erased signs are thus summed out without an
exponential factor in the degree.
\end{proof}

To record successive transitions, write $\mathsf S(v)$ for a success or $\mathsf F(v,u)$ for a
rejection followed by deletion of $u$. A \emph{history} is a finite word in
these symbols; it omits all sampled signs. Given the fixed outer order, a
history determines the active set, the stack, and the vertex being filled.

\begin{lemma}[History weights]\label{lem:rep-history}
For every feasible history $h$ with $f(h)$ failures and final active set
$K(h)$, and every signing $s$ of $G[K(h)]$,
\[
\Pr\bigl(h,\ \text{current signing}=s\bigr)\le\pi^{f(h)}\,\Phi_{K(h)}(s).
\]
In particular $\Pr(h)\le\pi^{f(h)}$.
\end{lemma}

\begin{proof}
The empty history has mass one on the empty state, which equals
$\Phi_\emptyset$. Each success or failure satisfies the corresponding transition bound of
Lemma~\ref{lem:rep-kernel}; stack changes and the start of the next outer call
are deterministic. Induction proves the pointwise bound, and
$\sum_s\Phi_K(s)\le1$ proves the second.
\end{proof}

The lemma bounds joint, unnormalized probabilities. Conditioning on $h$
would remove the factor needed for the bound.

\subsection{Counting histories and proof of Theorem~\ref{thm:repair}}

We count these histories by representing each completed outer call as an
ordered rooted tree: each rejection is an edge to
the ensuing recursive call, and a node's final success closes its call. Root
labels are fixed by the outer order. A root has at most $\Delta$ choices for
each child. At a nonroot, the parent is an inactive neighbour, leaving at most
$\Delta-1$ choices for each child. To encode a child, use its index in a fixed neighbour
list, excluding the parent at nonroots. The
depth is at most $n$. Weighting each rejection by $z$, let $Y_h(z)$ overcount
nonroot trees of height at most $h$:
\begin{equation}\label{eq:rep-trees}
Y_0(z)=1,\qquad Y_{h+1}(z)=\frac{1}{1-(\Delta-1)\,z\,Y_h(z)} .
\end{equation}
All series have nonnegative coefficients, and completed executions are
coefficientwise dominated by
\begin{equation}\label{eq:rep-forest}
F_n(z)=\bigl(1-\Delta\,z\,Y_n(z)\bigr)^{-n}.
\end{equation}
This bound allows excess height and infeasible labels.

\begin{lemma}[Prefix padding]\label{lem:rep-padding}
For fixed $t$, every history of length $t$ injects into the completed forests
counted by~\eqref{eq:rep-forest}, with the same number of failures.
\end{lemma}

\begin{proof}
Append success symbols for the current stack, from the top downward, and then
for the remaining outer roots in order. Each pending call closes with no
further children, so the result is a formal completed forest within the same
height bound. Truncation after $t$ symbols recovers the prefix. The appended
successes serve only as an encoding and need not be feasible transitions.
\end{proof}

After $t$ attempts with $f$ failures the active set has $t-2f\le n$ elements,
so $f\ge(t-n)/2$. Let $a_f$ be the coefficient of $z^{f}$ in $F_n$. By
Lemmas~\ref{lem:rep-history} and~\ref{lem:rep-padding}, for $t\ge n$ and every
$z>\pi$ at which~\eqref{eq:rep-forest} converges,
\begin{equation}\label{eq:rep-count}
\Pr(\mathsf T>t)\le\sum_{f\ge(t-n)/2}a_f\,\pi^{f}
\le F_n(z)\,(\pi/z)^{(t-n)/2}.
\end{equation}
This bounds finite prefixes directly and does not presuppose termination.

\begin{proof}[Proof of Theorem~\ref{thm:repair}]
\emph{Slack.} Let $z_0=1/(4(\Delta-1))$, so $\pi=z_0(1+\epsilon)^{-2}$.
Induction in~\eqref{eq:rep-trees} gives $Y_h(z_0)\le2$, hence
$\Delta z_0Y_n(z_0)\le\Delta/(2(\Delta-1))\le3/4$ and $F_n(z_0)\le4^{n}$.
Then~\eqref{eq:rep-count} gives $\Pr(\mathsf T>t)\le4^{n}(1+\epsilon)^{-(t-n)}$;
summing its minimum with one over $t$ gives $\mathbb E\,\mathsf T=O(n/\epsilon)$.

\emph{Threshold.} Let $\pi=1/(4(\Delta-1))$, $z=\pi(1+1/(16n^{2}))$, and
$\theta=\arctan(1/(4n))$. Then $(\Delta-1)z=(4\cos^{2}\theta)^{-1}$, and
induction in~\eqref{eq:rep-trees} gives
\[
Y_h(z)=2\cos\theta\,\frac{\sin((h+1)\theta)}{\sin((h+2)\theta)},
\]
the inductive step being $2\cos\theta\sin((h+2)\theta)-\sin((h+1)\theta)
=\sin((h+3)\theta)$. For $0\le h\le n$ all angles lie in
$(0,(n+2)/(4n)]\subseteq(0,3/4]$, so the denominators are positive and
$Y_n(z)<2\cos\theta<2$. Hence
$\Delta zY_n(z)<\frac{\Delta}{2(\Delta-1)}\bigl(1+\frac1{16}\bigr)\le\frac{51}{64}
<\frac45$ and $F_n(z)<5^{n}$, which with~\eqref{eq:rep-count}
proves~\eqref{eq:rep-tail}. Since $\log(1+x)\ge x/(1+x)$, the bound
in~\eqref{eq:rep-tail} is at least one for at most
$2n\log5\,(16n^{2}+1)$ values of $t-n$, and its remaining terms sum to at most
$2(16n^{2}+1)$. Thus
$\mathbb E\,\mathsf T\le n+1+2n\log5\,(16n^{2}+1)+2(16n^{2}+1)<100\,n^{3}$.
The tail tends to zero, so the algorithm terminates almost surely, and the
output bounds were established above. The arithmetic assertion is proved in
Section~\ref{sec:rep-exact}.
\end{proof}

The argument has the shape of a witness-tree analysis for the algorithmic
local lemma~\cite{MoserTardos}, with squared Gram volume in place of a product
measure. The determinant average has a classical interpretation:
by the Godsil--Gutman identity used in~\cite{MSS1}, $\sum_s2^{-|E|}\det(I-A_s/r)=r^{-n}\mu_G(r)$,
where $\mu_G$ is the matching polynomial. The weight~\eqref{eq:rep-weight}
retains only the positive definite terms of this average. The tree
recursion~\eqref{eq:rep-trees} converges exactly when $z\le1/(4(\Delta-1))$,
which in the one-sided mode is the deletion factor at the Heilmann--Lieb radius
$2\sqrt{\Delta-1}$ bounding the roots of $\mu_G$.

\subsection{Exact implementation}\label{sec:rep-exact}

In the two-sided mode all decisions are rational. Put
$M_K=r^{2}I-A_K^{2}$; at the threshold $r^{2}=8(\Delta-1)$ this is an integer
matrix with entries bounded by $r^{2}+\Delta$. The state is good exactly when
$M_K\succ0$, and $W_K=\det M_K/r^{2|K|}$. Consequently
\[
\alpha=\frac{\det M_{K+v}}{r^{2}\det M_K}\quad\text{if }M_{K+v}\succ0,
\qquad \lambda_u\propto(M_K^{-1})_{uu}.
\]
Fraction-free elimination tests definiteness and computes the determinants
and diagonal cofactors with polynomial bit length. Rational Bernoulli and
weighted choices use polynomially many random bits in expectation. Recomputing from scratch costs $O(n^{3})$ field operations per
attempt. With slack, rational $r^{2}$ gives the same conclusion, with
dependence on its bit length.

In the one-sided mode $Q_K^{-}$ has entries in $\mathbb Q(\sqrt{\Delta-1})$;
exact arithmetic and sign tests in this fixed quadratic field have polynomial
cost. For bipartite $G$ with classes $L,R$, write
$A_K=\bigl(\begin{smallmatrix}0&B\\B^{*}&0\end{smallmatrix}\bigr)$ and $\varsigma=r^{2}$.
Then $Q_K^{-}\succ0$ if and only if $\varsigma I-BB^{*}\succ0$,
$\det Q_K^{-}=\det(\varsigma I-BB^{*})/\varsigma^{|L\cap K|}$, and the diagonal
blocks of $(Q_K^{-})^{-1}$ are $\varsigma(\varsigma I-BB^{*})^{-1}$ and
$\varsigma(\varsigma I-B^{*}B)^{-1}$. At the threshold $\varsigma=4(\Delta-1)$
is an integer, and every decision is rational.

\begin{remark}[Verification]
The supplementary programs~\cite{certificate} check
Lemmas~\ref{lem:rep-schur}--\ref{lem:rep-kernel} exhaustively in exact
arithmetic on all signed exteriors of size $2\times3$ at $r^{2}\in\{4,8\}$.
They verify Lemmas~\ref{lem:rep-history} and~\ref{lem:rep-padding} on all
execution prefixes of $K_{2,2}$ through ten attempts. They also run the
one-sided bipartite and rational two-sided algorithms to completion, with
exact certificates for their outputs. These finite checks test the algebra and the implementation; the
general claims rest on the proofs above.
\end{remark}

\subsection{Consequences for lifts}

As in Section~\ref{sec:lifts}, a signing $s$ of $G$ defines a $2$-lift whose
new eigenvalues are those of $A_s$. Theorem~\ref{thm:repair} therefore
strengthens Corollary~\ref{cor:lifts} for the adjacency spectrum, at the cost
of randomization.

\begin{corollary}[Randomized lift families]\label{cor:rep-lifts}
From a $d$-regular graph $G_0$ on $n_0$ vertices, $d\ge3$, one can construct a
sequence of $2$-lifts $G_k$ on $n_02^{k}$ vertices, in expected time polynomial
in the output size, such that
\[
\lambda(G_k)\le\max\bigl\{\lambda(G_0),\,2\sqrt{2(d-1)}\bigr\}.
\]
For bipartite $G_0$, replace the second term by $2\sqrt{d-1}$. If $G_0$ is
connected, so are all lifts, with a uniform spectral gap for $d\ge7$
(bipartite: $d\ge3$).
\end{corollary}

\begin{proof}
Apply Theorem~\ref{thm:repair} at each level, in the two-sided mode, or in the
one-sided mode for bipartite graphs, whose lifts remain bipartite. The
thresholds follow from $8(d-1)<d^{2}$ for $d\ge7$ and $4(d-1)<d^{2}$ for
$d\ge3$. The expected costs $O(n_0^{3}2^{3j})$ attempts at level $j$ sum to a
polynomial in $n_02^{k}$.
\end{proof}

Compared with Corollary~\ref{cor:lifts}, the degree thresholds fall from $91$
to $7$ and from $46$ to $3$.  If the prescribed bipartite base $G_0$ is Ramanujan, the eigenvalues
inherited from $G_0$ satisfy the required bound. Theorem~\ref{thm:repair} bounds the
new eigenvalues, so induction shows that every $G_k$ is Ramanujan. The
algorithm thus constructs a good $2$-lift of the supplied graph itself.  By
contrast, Cohen~\cite{Cohen} constructs bipartite Ramanujan multigraphs from
size and degree through a matching interlacing family.  For general bases the
same repair process gives lift families whose new eigenvalues are within a
factor $\sqrt2$ of the Ramanujan threshold.

\section{Discussion and conclusion}\label{sec:discussion}
The proof of the deterministic rounding theorem has two parts. The matrix argument
identifies the responses to coordinate directions and constructs a response
frame with the required covariance. The local inequality shows that the
negative second-derivative term from the weight function dominates the
response terms. Rank
enters through trace inequalities for positive matrices. We can therefore
assign one sign to a higher-rank input without splitting it into pieces that
receive independent signs.

The local inequality is the computer-assisted part of the proof. A subdivision
certificate verifies it with exact integer arithmetic and a strict margin.
The certificate is independent of the input matrices. During rounding, the
algorithm uses the resulting formulas for the weight function and metric; it does not
search for new formulas. The analytic proof then shows that rounding makes
progress at the chosen precision.

These bounds on progress and precision establish polynomial bit complexity
for deterministic rounding. They
cover matrix inversions, linear solves, factorizations, and certified
comparisons. The resulting iteration bound is conservative.  An end-to-end certified solver for the deterministic rounding algorithm
remains to be implemented.  The supplementary programs~\cite{certificate} verify the scalar certificate and supporting
algebra. They also test the exact repair algorithms of
Section~\ref{sec:repair} on finite instances.

A separate limitation concerns the trace scale in the rounding guarantee:
at rank one, it equals matrix variance; at higher rank, it can be larger.

The two algorithms also differ in which matrices they control. The rounding algorithm is
deterministic and keeps higher-rank summands intact. For graphs, it controls
signed degrees and adjacency together. The repair algorithm of
Section~\ref{sec:repair} is randomized and controls adjacency alone. Its
constant is within $\sqrt2$ of the conjectured bound and attains the sharp
universal bound for bipartite graphs. Both its two-sided mode and its
one-sided mode on bipartite graphs use exact rational comparisons. The
analysis needs neither a spectral potential nor finite-precision
certification.

For the repair algorithm, three questions remain. First, the factor $\sqrt2$ enters only through the
bound $\delta_u^++\delta_u^-\le2$ in Lemma~\ref{lem:rep-charge}. A deletion rule whose bound involved one leave-one-out residual instead of
their sum would give the conjectured constant. Second, Marcus, Spielman, and Srivastava prove the graph-specific
bound $\lambda_{\max}(A_s)\le\rho(\mu_G)$, which can lie well below
$2\sqrt{\Delta-1}$. In the one-sided mode, the total mass of the repair
weights is the positive-definite part of the Godsil--Gutman average
$r^{-n}\mu_G(r)$. This suggests that a repair rule adapted to $G$ might reach
that bound; the degree-based tree count used here cannot. Third, the
determinantal weights of the histories suggest a possible derandomization by
conditional expectations. We have not pursued this. Our branching count uses
the graph degree, and we do not know an analogue for general rank-one or
higher-rank discrepancy.

\subsection*{Statement on AI Usage}
This work is the product of extensive back-and-forth interactions with AI, and substantial human input. Both Claude Fable 5.1 as well as GPT-5.6 and GPT-6 (Pro versions) have been used for all aspects of the paper preparation as well as for helping generate verification code and certificates for the technical material.

\appendix

\section{Regularity and response identities}\label{sec:response}
This appendix supplies the analytic facts used in the main proof. Coordinate
sums run over the active inputs, while $S$ includes frozen contributions.
By Lemma~\ref{lem:trace}, the trace map is positive and self-adjoint and has
$\norm{\Eop_x}_{F\to F}\le c$.

\subsection{The optimizer and its linearization}
\begin{lemma}[Primal--dual data]\label{lem:kkt}
On every open face the minimizer of~\eqref{eq:potential} is unique and smooth. Its unique dual pair $P,Q\succ0$ satisfies
\begin{equation}\label{eq:kkt}
 \begin{gathered}
 X^{-1}+S+\Eop_x(Y)=tI,\qquad Y^{-1}-S+\Eop_x(X)=tI,\\
 P=X\Eop_x(Q)X+\rho X^2,\qquad
 Q=Y\Eop_x(P)Y+\rho Y^2,\qquad \Tr(P+Q)=1.
 \end{gathered}
\end{equation}
At accepted states $R<4$, so $\norm{X^{-1}},\norm{Y^{-1}}<8$, $\norm X,\norm Y\le\rho^{-1/2}$, and $P,Q\succeq\rho I/64$. The response operator $\Lop$ in~\eqref{eq:L} is invertible, its inverse preserves positive semidefiniteness, and $\norm{\Lop^{-1}}_{F\to F}\le\rho^{-1}$.
\end{lemma}
\begin{proof}
Schur complements turn the constraints into linear matrix inequalities. A scalar test point gives strict feasibility. On a bounded sublevel of the objective, the trace penalty bounds $X,Y$ from
above, and the constraints bound them away from zero. The minimum therefore
exists. Slater duality gives~\eqref{eq:kkt}. The stationarity equations make
$P,Q$ positive definite, so complementary slackness makes both constraints
tight.

With the optimal multipliers fixed, the inversion terms in the Lagrangian
are strictly convex. The primal optimum is therefore unique. Taking traces in stationarity gives $\rho\Tr(X^2+Y^2)\le1$. The remaining size bounds follow from $\norm S\le R<4$ and $t\le R$.

Let $\mathcal T(U,V)=(X\Eop_x(V)X,Y\Eop_x(U)Y)$. Stationarity gives
$\mathcal T(P,Q)=(P,Q)-\rho(X^2,Y^2)\preceq(1-\rho/64)(P,Q)$.
Since $(P,Q)$ is an order unit, the associated Neumann series converges and
yields a positive inverse of $\Lop$. Moreover,
$\Lop^{-1}(I,I)=(P,Q)/\rho\preceq(I,I)/\rho$. Thus its order norm is at most
$\rho^{-1}$, and the same bound holds for its spectral radius.
Self-adjointness gives the Frobenius norm bound and uniqueness of the dual
pair.

For smoothness let $\mathscr H$ be the Hessian of the primal Lagrangian in the matrix variables. Its first component is
$X^{-1}UX^{-1}PX^{-1}+X^{-1}PX^{-1}UX^{-1}$, with the analogous $Y,Q$ component. Stationarity implies $\mathscr H\succeq2\rho^{3/2}I$. The KKT linearization reduces to
\begin{equation}\label{eq:jacobian}
 \Lop U+\lambda\mathbf e=f_1,\qquad
 \Lop V-\mathscr H U=f_2,\qquad
 \ip{\mathbf e}{V}=f_3,\qquad \mathbf e=(I,I).
\end{equation}
Its scalar Schur complement is
$\ip{\Lop^{-1}\mathbf e}{\mathscr H\Lop^{-1}\mathbf e}
\ge1/(d\sqrt\rho)>0$, using $\Tr(P^2+Q^2)\ge1/(2d)$.
The Jacobian is invertible, and the implicit function theorem applies. The same calculation gives uniform polynomial bounds on each fixed sublevel
of the objective; see Appendix~\ref{sec:precision}.
\end{proof}

\subsection{Two response components per coordinate}
Write $\psi_i=\psi(x_i)$, $\psi_i^\prime=\psi^\prime(x_i)$, and similarly for higher derivatives. At a light state define $a_i=\Tr(M_iX)$, $b_i=\Tr(M_iY)$, $p_i=\Tr(M_iP)$, and $q_i=\Tr(M_iQ)$. Put
\begin{equation}\label{eq:localscales}
 \begin{gathered}
 \alpha_i=1+c\psi_i'b_i,\quad \beta_i=1-c\psi_i'a_i,\quad
 t_i=\sqrt{\alpha_i\beta_i/\psi_i},\\
 \ell_i=\sqrt{\psi_i\alpha_i/\beta_i},\quad
 k_i=\sqrt{\psi_i\beta_i/\alpha_i},\quad
 \widetilde a_i=\ell_i a_i,\quad \widetilde b_i=k_i b_i,\quad
 z_i=\ell_i p_i,\quad w_i=k_i q_i.
 \end{gathered}
\end{equation}
The bounds on the weight function in Appendix~\ref{sec:certificate} imply $\alpha_i,\beta_i>0$. Let
\[
 \mathsf A_{ij}=c\ell_i\ell_j\Tr(M_iXM_jX),\qquad
 \mathsf B_{ij}=ck_i k_j\Tr(M_iYM_jY),
\]
and $D_F=\diag(\widetilde a_i z_i)$, $D_G=\diag(\widetilde b_i w_i)$, $D_0=D_F+D_G$. For a direction $v$, solve
\begin{equation}\label{eq:coeffresponse}
 F-\mathsf B G=\diag(t_i)v,\qquad
 G-\mathsf A F=-\diag(t_i)v.
\end{equation}
Then $\dot X=XK_XX$, $\dot Y=YK_YY$, where $K_X=\sum_i\ell_iF_iM_i$ and $K_Y=\sum_i k_iG_iM_i$, are the responses of the tight constraints with $t$ held fixed. Indeed~\eqref{eq:L} sends this pair to $(\sum_i\alpha_iv_iM_i,-\sum_i\beta_iv_iM_i)$. The homogeneous coefficient system is injective by the same argument and Lemma~\ref{lem:kkt}, even if the inputs are linearly dependent.

With $y=(D_F^{1/2}F,D_G^{1/2}G)$ and $\xi=D_0^{1/2}\diag(t_i)v$, equation~\eqref{eq:responsecore} holds with
\begin{equation}\label{eq:TJ}
 \mathsf T=\begin{pmatrix}0&D_F^{1/2}\mathsf B D_G^{-1/2}\\
 D_G^{1/2}\mathsf A D_F^{-1/2}&0\end{pmatrix},\qquad
 J=\binom{D_F^{1/2}D_0^{-1/2}}{-D_G^{1/2}D_0^{-1/2}}.
\end{equation}
We interleave the coordinates into $m$ pairs when referring to a local $2\times2$ block. Define the response-energy matrix $\mathsf E\succeq0$ by
\begin{equation}\label{eq:energy}
 y^*\mathsf E y=\Tr(PK_XXK_X)+\Tr(QK_YYK_Y).
\end{equation}

\begin{lemma}[Rank-free energy bounds]\label{lem:purity}
For each active coordinate there are numbers $f_i,g_i,u_i,v_i$ satisfying
$0<f_i,g_i\le1$, $f_i^2\le u_i\le1$, and $g_i^2\le v_i\le1$, such that
\begin{equation}\label{eq:purities}
 \mathsf E_{ii}=\diag(f_i,g_i),\quad
 \mathsf A_{ii}=c\widetilde a_i^2u_i,\quad
 \mathsf B_{ii}=c\widetilde b_i^2v_i,\quad
 \mathsf E\succeq\mathsf T^*\mathsf C\mathsf T,
 \qquad \mathsf C=\bigoplus_i\frac{I_2}{c\widetilde a_i\widetilde b_i}.
\end{equation}
\end{lemma}
\begin{proof}
Normalize $M_i^{1/2}XM_i^{1/2}$ and $M_i^{1/2}PM_i^{1/2}$ to trace-one matrices $H_i,L_i$. Take $f_i=\Tr(H_iL_i)$ and $u_i=\Tr(H_i^2)$, and use $Y,Q$ for $g_i,v_i$. Lemma~\ref{lem:density} gives $f_i^2\le u_i\le1$ and the analogous bounds for $g_i,v_i$. Substitution in~\eqref{eq:energy} gives the diagonal identities.

For the last inequality, $P\succeq c\sum_j\psi_jq_jXM_jX$. Apply
$|\Tr(MD)|^2\le\Tr(MX)\Tr(MDX^{-1}D)$ with $D=XK_XX$. This yields
\[
 \Tr(PK_XXK_X)\ge
 \sum_j\frac{(D_G)_{jj}(\mathsf A F)_j^2}{c\widetilde a_j\widetilde b_j}.
\]
The $Y$ inequality supplies the other block. The algorithm never factors $M_i$.
\end{proof}

\begin{lemma}[Second variation with the mixed term retained]\label{lem:variation}
The gradient is $\partial_iR=t_i(z_i-w_i)$. If $\mathsf N$ has disjoint local columns
\[
 n_i=\frac{\psi_i'/\psi_i}{t_i\sqrt{(D_0)_{ii}}}
 \binom{z_i/\sqrt{(D_F)_{ii}}}{w_i/\sqrt{(D_G)_{ii}}},
\]
then every exact response satisfies
\begin{equation}\label{eq:variation}
 \frac12\nabla^2R[v,v]\le y^*\mathsf E y+\xi^*\mathsf N^*y
 -\sum_i\frac{\psi_i'}{\psi_i}(\partial_iR)v_i^2
 +\frac c2\sum_i\psi_i''(b_ip_i+a_iq_i)v_i^2.
\end{equation}
\end{lemma}
\begin{proof}
Hold $t$ fixed and solve the two tight constraints along $x+sv$. Their objective $\varphi(s)$ majorizes $R(x+sv)$ and agrees with it at zero. Differentiate the constraints and pair them with $P,Q$. Stationarity cancels the first derivatives of $X,Y$, giving the gradient formula. Differentiating again and using the same cancellation gives
\[
 \tfrac12\varphi''(0)=y^*\mathsf E y
 +c\sum_i\psi_i'v_i\{p_i\Tr(M_i\dot Y)+q_i\Tr(M_i\dot X)\}
 +\tfrac c2\sum_i\psi_i''(b_ip_i+a_iq_i)v_i^2.
\]
Use $ck_i\Tr(M_i\dot Y)=F_i-t_iv_i$ and $c\ell_i\Tr(M_i\dot X)=G_i+t_iv_i$. The middle term becomes the mixed term and gradient term in~\eqref{eq:variation}. Finally $\varphi''(0)\ge\nabla^2R[v,v]$.
\end{proof}

\section{Proof of the completion identity}\label{sec:completion}
We prove the completion lemma stated in Section~\ref{sec:descent}.
That section describes the normalization and the covariance of the resulting directions. Here we expand the two squares that
establish~\eqref{eq:completionconclusion}.

\begin{proof}[Proof of Lemma~\ref{lem:completion}]
Scale $W_i$ by $\chi_i^{-1/2}$ and $D_i$ by $\chi_i^{-1}$; this divides $h_i,U_i,R_i$ by $\chi_i$ and leaves $K_i$ unchanged. Assemble block-diagonal matrices and put
$P_0=JJ^*$, $\mathsf M=(I-\mathsf T)W$. Let $\Pi$ project onto
$\ker((I-P_0)\mathsf M)$, which has dimension $m$, and set $\Pi^\perp=I-\Pi$.
Write $B=D^*W^{-1}$, $U=W^{-1}DW$, $R=B^*\mathsf C^{-1}B/4$,
$Z_1=P_0B+J \mathsf N^*W$, and $K=\bigoplus_iK_iI_2$.
Expanding two Frobenius squares gives the identity
\begin{align}\label{eq:squares}
 &\norm{K^{1/2}\mathsf M\Pi}_F^2
 +\norm{\mathsf C^{1/2}(\mathsf M-W)\Pi^\perp}_F^2
 +\Tr(D\mathsf T)-\Tr(W \mathsf N J^*\mathsf M\Pi)\notag\\
 &=\norm{K^{1/2}\mathsf M\Pi-\tfrac12K^{-1/2}Z_1\Pi}_F^2
 +\norm{\mathsf C^{1/2}(\mathsf M-W)\Pi^\perp
              -\tfrac12\mathsf C^{-1/2}B\Pi^\perp}_F^2\notag\\
 &\hspace{12mm}+\Tr\bigl((\Sym U+R-\tfrac14Z_1^*K^{-1}Z_1)\Pi\bigr)-\Tr R.
\end{align}
For the cross terms use $\mathsf M\Pi=P_0\mathsf M\Pi$, $WB=U^*$, and $\Tr(\mathsf M^*B)=\Tr D-\Tr(D\mathsf T)$.
The last condition in~\eqref{eq:completionconditions}, together with the $2\times2$ adjugate identity, says $\Sym U+R-Z_1^*K^{-1}Z_1/4\succeq I$. Thus the last line of~\eqref{eq:squares} is at least $m-\Tr R$, the normalized sum of the $h_i$.

The trace condition bounds $\Tr(W\mathsf EW)+\Tr(D\mathsf T)$ by that same sum. Also
$\Tr(W\mathsf EW\Pi^\perp)\ge\norm{\mathsf C^{1/2}(\mathsf M-W)\Pi^\perp}_F^2$.
Subtracting and discarding the nonnegative squares proves
$\Tr(W\mathsf EW\Pi)+\Tr(W \mathsf N J^*\mathsf M\Pi)\le\norm{K^{1/2}\mathsf M\Pi}_F^2$.
For an isotropic vector $g\in\R^{2m}$ take $y=W\Pi g$ and $\xi=J^*\mathsf M\Pi g$. These are exact responses, and the last inequality is~\eqref{eq:completionconclusion}. The covariance is nonzero since $\mathsf M$ is invertible and $\ran\Pi$ has dimension $m$.
\end{proof}

\section{The scalar certificate and the approximation constant}\label{sec:certificate}
This appendix specifies the weight function and metric used in Theorem~\ref{lem:core}.
Throughout, $c=567/200$ and $\varpi=1/50000$. The formulas for the weight function and metric are fixed before rounding starts; verifying the scalar
inequality they must satisfy is the computer-assisted part of the proof.

\subsection{The weight function and its shape}
For $z=x^2$, $r=1-z$, and $s=\sqrt r$, set $\psi(x)=sP_b(z)$, where
\begin{equation}\label{eq:profile}
 \begin{split}
 P_b(z)={}&1-0.04912z-0.05594z^2-0.02446z^3\\
          &-0.19169z^4+0.36085z^5-0.32665z^6.
 \end{split}
\end{equation}
Every terminating decimal specifying a coefficient in this paper is an exact rational. Define
\[
 D_b=P_b-2rP_b',\qquad
 N_b=P_b-2(1-3z)rP_b'-4zr^2P_b''.
\]
Then $\psi'=-xD_b/s$ and $-\psi''=N_b/s^3$. The exact univariate checks give, on $[0,1]$,
\begin{equation}\label{eq:shape}
 \begin{gathered}
 3/5\le P_b,D_b,N_b\le2,\quad P_b\le1,\quad P_b'\le0,\quad P_b^2\ge1-z,\\
 P_b^2\ge16zr^2(P_b')^2,\quad
 \tfrac32 V_b\ge N_bP_b,\quad
 2P_bN_b\ge zD_b^2,\quad
 8V_b^2\ge4N_bP_b^3,\qquad V_b=P_b^2-4zr^2(P_b')^2.
 \end{gathered}
\end{equation}
In particular $1-x^2\le\psi(x)\le\sqrt{1-x^2}$ and $\psi$ is strictly concave. We certify these statements by checking polynomial positivity. On an
interval, expand $Q(z)=\sum a_kz^k$ in the Bernstein basis, with coefficients
$b_j=\sum_{k\le j}a_k\binom jk/\binom nk$. Nonnegative coefficients imply
$Q\ge0$. The supplied verifier reconstructs the polynomials
from~\eqref{eq:profile}. Only the lower bound on $N_b$ and the bound involving
$zD_b^2$ require dyadic subdivision, into four and two intervals respectively.

Since $(1-x)\psi^{\prime}(x)\to0$ at the endpoint, integration by parts in Section~\ref{sec:potential} is valid.

\subsection{A single light coordinate}
Suppress its index and first assume $0\le x<1$. In this subsection $a,b$ denote $\Tr(M_iX),\Tr(M_iY)$; all other scalar quantities are local. Put
\begin{equation}\label{eq:lightbox}
 \begin{gathered}
 U=c\psi a,\quad V=c\psi b,\quad
 \delta=-2xrP_b'/P_b,\quad q=x+\delta=xD_b/P_b,\\
 A=\frac{U}{r+qU},\quad B=\frac{V}{r-qV},\quad
 \phi=\frac A{A+B},\quad {\bar\phi}=1-\phi,\quad\theta=\phi{\bar\phi}.
 \end{gathered}
\end{equation}
Lightness is exactly $0<A<(1+\delta)^{-1}$ and $0<B<(1-\delta)^{-1}$.
Also $\alpha=(1+qB)^{-1}$, $\beta=(1-qA)^{-1}$, and $\phi=\widetilde a/(\widetilde a+\widetilde b)$. The shape bounds imply $0\le\delta\le1/2$. Define
\begin{equation}\label{eq:certificatebase}
 \begin{gathered}
 p=(N_b/P_b)AB,\quad B_0=(N_b/P_b)(A+B)^2=p/\theta,\quad
 \zeta=zD_b^2/(P_bN_b),\quad J_0=q(A+B),\\
 \kappa=\frac{cN_b}{2s^3}(1-qA)(1+qB),\quad
 C=2\kappa/p=(c\widetilde a\widetilde b)^{-1},\quad
 m_0=2\sqrt{\zeta/p}.
 \end{gathered}
\end{equation}
Here $p\le3/2$, $B_0\le8$, $\zeta\le2$, and $J_0^2=\zeta B_0$. At $z_i=w_i$, the last two terms of~\eqref{eq:variation} sum to $-\kappa\xi_i^2$.

In the orthonormal forcing/mixed basis
$j=(\sqrt\phi,-\sqrt{{\bar\phi}})^*$, $k=(\sqrt{{\bar\phi}},\sqrt\phi)^*$,
the columns in Lemma~\ref{lem:variation} are $(1,0)^*$ and $(0,-\kappa m_0)^*$. Choose
\begin{equation}\label{eq:multipliercoefficients}
 \begin{gathered}
 u=1.91235+pz(-0.89868+0.31535z),\qquad
 v=1.00261+pz(-0.30669-0.03467z),\\
 e=0.60598+0.20181z,\quad \eta_0=1-pe,\quad
 \omega=1.10948-0.76493z,\quad w=\theta\omega,\\
 R_0=2u(\eta_0-1)-4\eta_0-u\eta_0p,\qquad
 \gamma=R_0/(4+vB_0/\omega),\\
 W_0=\diag(1,\sqrt w),\qquad
 D=\kappa\begin{pmatrix}u&m_0\gamma\\\eta_0m_0w&v\end{pmatrix},
 \qquad B_D=D^*W_0^{-1}.
 \end{gathered}
\end{equation}
All coefficients in this subsection are local. Before normalization, the matrices used in the completion lemma are $[j,k]W_0[j,k]^*$ and $[j,k]D[j,k]^*$.

\subsection{The local matrix inequality}
Using the support function $\mathcal H$ from Lemma~\ref{lem:density}, set
\begin{equation}\label{eq:tracecharge}
 \begin{aligned}
 A_f&=\phi+{\bar\phi}w,& A_g&={\bar\phi}+\phi w,\\
 e_f&=\tfrac12\phi(u-v)p-\phi^2J_0\gamma+\eta_0J_0\theta w,&
 e_g&=\tfrac12{\bar\phi}(u-v)p+{\bar\phi}^2J_0\gamma-\eta_0J_0\theta w,\\
 h&=\mathcal H(A_f,e_f)+\mathcal H(A_g,e_g).
 \end{aligned}
\end{equation}
At balance, the local interaction block in the original basis is
$C^{-1}\left(\begin{smallmatrix}0&\sqrt{{\bar\phi}/\phi}\,v_i\\
\sqrt{\phi/{\bar\phi}}\,u_i&0\end{smallmatrix}\right)$.
The local trace bound is
$A_ff_i+A_gg_i-e_fu_i-e_gv_i\le h$, by direct multiplication and
Lemma~\ref{lem:purity}. We use both support functions and do not assume that either purity attains
its rank-one value.

The mixed row is $d^*=\kappa(u,(\eta_0-1)m_0\sqrt w)$. With coefficient $K_i=\kappa$, the residual in~\eqref{eq:completionconditions} is diagonal:
\begin{equation}\label{eq:residual}
 \Sym(W_0^{-1}DW_0)-hI_2
 -\adj\!\left(\frac{B_D^*B_D}{4C}\right)-\frac{dd^*}{4\kappa}
 =\diag(\kappa F_A-h,\kappa F_B-h),
\end{equation}
where
\begin{equation}\label{eq:FAFB}
 \begin{split}
 F_A&=u-u^2/4-\eta_0^2\zeta w/2-v^2B_0/(8\omega),\\
 F_B&=v-u^2p/8-\zeta\gamma^2/2-e^2p\zeta w.
 \end{split}
\end{equation}
To verify the identity, expand the two quadratic terms
in~\eqref{eq:residual} and use $\kappa/C=p/2$ and $m_0^2=4\zeta/p$. The off-diagonal entry vanishes because $\gamma=wR_0/(4w+pv)$. The remaining diagonal entries are~\eqref{eq:FAFB}. A separate symbolic check verifies these identities for unrestricted
parameters. The interval verification does not enter that check.

\begin{lemma}[Uniform scalar inequality]\label{lem:scalar}
For the exact coefficients above, throughout the light domain,
\begin{equation}\label{eq:certified}
 cF_A-\frac{h}{\kappa/c}>\frac1{2000},\qquad
 cF_B-\frac{h}{\kappa/c}>\frac1{2000}.
\end{equation}
In particular $\kappa F_A-h,\kappa F_B-h>1/20000$.
Also $1/128\le h\le34$, $\kappa>1/4$, and $\norm d^2<7\kappa^2$.
\end{lemma}
\begin{proof}[Exact certificate and elementary bounds]
The subdivision certificate and integer verifier (the files named \texttt{round2}) will be made available on GitHub~\cite{certificate}. Rationalize $x=2q_0/(1+q_0^2)$ and $s=(1-q_0^2)/(1+q_0^2)$. Write
$\widehat A=A(1+\delta)$, $\widehat B=B(1-\delta)$ and use the two charts
$(\widehat A,\widehat B)=(M_0,M_0z_0),(M_0z_0,M_0)$.
After cancelling the common factor in $\phi$, these formulas extend to $M_0=0$. The two cubes $(q_0,M_0,z_0)\in[0,1]^3$ cover the entire domain.

The verifier reconstructs every box from the two roots and checks~\eqref{eq:certified} with outward-rounded integer intervals. Its denominator is $2^{48}$; division and overflow checks are part of acceptance. Every leaf must satisfy the stated strict margin, and both roots must be exhausted. The generator's choices play no role in acceptance. The planned release~\cite{certificate} will also include verification statistics and cross-checks.

The verification also covers the endpoint. Put
$d_A=(1-\widehat A)+\widehat A(1-x)/(1+\delta)$, so
$\kappa/c=N_b(1+qB)d_A/(2s^3)$. The inequality $d_A\ge(1-x)/(1+\delta)$ gives
\begin{equation}\label{eq:endpointcertificate}
 \frac1{\kappa/c}\le
 \frac{2(1-q_0)(1+q_0)^3(1+\delta)}
 {(1+q_0^2)^2N_b(1+qB)}.
\end{equation}
This upper bound vanishes continuously at $q_0=1$. The verifier uses it together with the direct quotient where that quotient has a positive denominator. Midpoint-translated Horner evaluation bounds the polynomial ranges. The independent univariate verifier supplies~\eqref{eq:shape}.

For the remaining claims,~\eqref{eq:shape} gives $\kappa/c\ge1/10$, which converts~\eqref{eq:certified} to the stated residual margin. The explicit coefficients imply
$1/2\le u\le2$, $2/5\le v\le11/10$, $0<e<81/100$, and $1/4\le\omega\le5/4$.
Hence $w\le5/16$, $|\eta_0|\le1$, $|R_0|<12$, $|\gamma|<3$, and $|e_f|,|e_g|<16$. Since $A_f+A_g\ge1$ and $A_f,A_g\le1$, testing $f=A_f/32$ and $g=A_g/32$ in~\eqref{eq:support} gives $h\ge(A_f^2+A_g^2)/64\ge1/128$; $h\le34$ is immediate. Finally
$\norm d^2/\kappa^2=u^2+4e^2p\zeta w<7$.
\end{proof}

\subsection{From balance to the actual state}\label{sec:robustness}
The strict margin above lets us perturb away from exact balance. Lower
the coefficient to $K_i=\kappa_i-\varpi$ and increase the trace bound to $h_i=h+\varpi/10$. The mixed quadratic term increases by at most
$7\kappa_i\varpi/[4(\kappa_i-\varpi)]<2\varpi$.
The balanced residual therefore remains above $1/20000-2.1\varpi=8\cdot10^{-6}$. With $\chi_i=h_i+\norm{B_D}_F^2/(4C_i)$, the algorithm's metric is
\begin{equation}\label{eq:normalizedmetric}
 W_i=\chi_i^{-1/2}[j,k]W_0[j,k]^*.
\end{equation}
In the proof of the completion lemma, the multiplier $D_i$ is simultaneously
divided by $\chi_i$. This multiplier enters only the proof; the algorithm
does not use it to select directions.

To check robustness, fix a coordinate's primal data and local multipliers, and set $a=\widetilde a_i$, $b=\widetilde b_i$, $\varrho=w_i/z_i$, and $d_0=(\psi_i'/\psi_i)/t_i$. In the original pair basis the actual forcing and mixed columns, and the actual curvature coefficient, are
\begin{equation}\label{eq:actualratio}
 \begin{gathered}
 j(\varrho)=\frac{(\sqrt a,-\sqrt{b\varrho})^*}{\sqrt{a+b\varrho}},\qquad
 n(\varrho)=\frac{d_0(a^{-1/2},\sqrt\varrho\,b^{-1/2})^*}{\sqrt{a+b\varrho}},\\
 \kappa_b(\varrho)=\kappa\frac{b+a\varrho}{a+b\varrho}
                 +d_0\frac{1-\varrho}{a+b\varrho},\qquad
 \mathsf T_{ii}(\varrho)=\frac1C
 \begin{pmatrix}0&\sqrt{b/a}\,\varrho^{-1/2}v_i\\
 \sqrt{a/b}\,\varrho^{1/2}u_i&0\end{pmatrix}.
 \end{gathered}
\end{equation}
The last two terms of~\eqref{eq:variation} are $-\sum_i\kappa_b(\varrho_i)\xi_i^2$. Put $n_0=|d_0|/\sqrt{ab}$ and $T_0=C^{-1}\max(\sqrt{a/b},\sqrt{b/a})$. On $[1/2,3/2]$, direct differentiation gives
$\norm{j'}\le1$, $\norm{n'}\le2n_0$, $\norm{\mathsf T_{ii}'}_F\le2T_0$, and $|\kappa_b'|\le4\kappa+4n_0$.
Thus a common Lipschitz bound for the trace expression, matrix residual, and curvature coefficient is
\[
 L_i=1+4\norm{D}_FT_0+
 \frac{4(\norm{B_D}_F+2\norm{W_0}_Fn_0)^2}{K_i}
 +4\kappa+4n_0.
\]
With these fixed multipliers choose
$\zeta_i=\min(1/2,\varpi/(100L_i))$ and
$g_0(x)=\tfrac12\min_i t_i\min(z_i,w_i)\zeta_i$.
Since $\partial_iR=t_iz_i(1-\varrho_i)$, a gradient below $g_0(x)$ gives $|1-\varrho_i|\le\zeta_i/2$. The resulting perturbation is at most $\varpi/200$, leaving the completion inequalities valid and $\kappa_b(\varrho_i)-K_i\ge\varpi/2$. We use the purity bounds uniformly and do not differentiate the purities.

Apply Lemma~\ref{lem:completion} and normalize its covariance as
in~\eqref{eq:inducedcovariance}. Because $t_i^2\ge1/4$, we obtain
\begin{equation}\label{eq:actualcurvature}
 \frac1{2m}\sum_j\nabla^2R[h^{(j)},h^{(j)}]
 \le-\frac\varpi8\min_i(D_0)_{ii}
 \le-a_x,\qquad a_x=\frac\varpi{10}\min_i(D_0)_{ii}.
\end{equation}
For $x_i<0$ exchange its two response components and replace $(x_i,v_i)$ by $(-x_i,-v_i)$. The even weight function, response equation, and quadratic form are invariant under this orthogonal relabeling. The same certificate applies. The uniform bounds on $g_0,a_x$ in Appendix~\ref{sec:precision} complete Theorem~\ref{lem:core}.

\subsection{Parameter choices and the resulting constant}\label{sec:constantaccount}
The strict margin above is independent of the input dimension. We now choose
small parameters to cover the errors from preprocessing, endpoint tests, and
rational scaling. With $M$ fixed at its value immediately after dependence
elimination, take
\begin{equation}\label{eq:parameters}
 \rho=\frac1{10^4d},\qquad
 \tau=\sigma=\delta_E=\frac1{10^4M},\qquad
 V\le b^2\le(1+10^{-6})V.
\end{equation}
The three coordinate error allowances sum to at most $3\cdot10^{-4}$.
Their exact sizes serve only to make the following quantitative statement
reproducible.

\begin{theorem}[Sharper quantitative form]\label{thm:sharp}
Under the hypotheses of Theorem~\ref{thm:main}, when $V>0$ the algorithm can
return a signing with discrepancy less than $\Cstar\sqrt V$, where
\begin{equation}\label{eq:exactconstant}
 \Cstar=3.367912113,\qquad \Czero=\sqrt{567/50}.
\end{equation}
Any rational $0<\epsilon\le1$ permits coefficient $\Czero+\epsilon$ in time polynomial also in $\epsilon^{-1}$.
\end{theorem}
\begin{proof}
The main proof gives~\eqref{eq:roundingaccount}. The choices above satisfy
\begin{equation}\label{eq:finalconstant}
 \left(2\sqrt{c+2d\rho}+M(\tau+\sigma+\delta_E)\right)
 \sqrt{1+10^{-6}}<\Cstar.
\end{equation}
This is an exact rational square comparison, checked by the supplementary programs~\cite{certificate}.
It also keeps every accepted normalized state below $4$.
For the limiting form, replace the fixed small error parameters by sufficiently
small rational multiples of $\epsilon$, and make the relative scaling error
quadratic in $\epsilon$. The initial coefficient tends to $2\sqrt c=\Czero$.
All conditioning and precision bounds in Appendix~\ref{sec:precision} remain
polynomial in the inverse parameters, giving the stated dependence.
\end{proof}

\Needspace{14\baselineskip}
\section{Optimization and finite precision}\label{sec:implementationdetails}
We give the quantitative estimates needed to prove that the main algorithm
runs in polynomial time. All constants here refer to the weight function of
Appendix~\ref{sec:certificate} and the parameters in~\eqref{eq:parameters}.

\subsection{Preprocessing and normalization}\label{sec:preprocessing}
We implement lines~\ref{line:preprocess}--\ref{line:initialize} of
Algorithm~\ref{alg:main}.
Identify each Hermitian input with a rational column in $\R^p$. Maintain an independent active list $B$ and its left Moore--Penrose inverse $D$. For a new column $a$, compute $u=Da$ and $r=a-Bu$. If $r\ne0$, append $a$ and update
\[
 D_{\rm new}=\binom{D-uw^*}{w^*},\qquad w^*=r^*/(r^*r).
\]
If $r=0$, the direction $(-u,1)$ preserves the matrix sum; move until some coordinate reaches an endpoint. If an old basis column $j$ freezes and the new one remains active, exchange them using
$D_{\rm new}=(I-(u-e_j)e_j^*/u_j)D$.
Deleting a frozen column changes every surviving row $d_i^*$ to
$d_i^*-\ip{d_i}{d_j}d_j^*/\norm{d_j}^2$.
These identities follow by checking the left-inverse and row-space
conditions. Each update costs $O(p^2)$ operations, and there are $O(N)$ events. The matrices in the active list remain original inputs. The dependence coefficients are rational functions of input submatrices.
Successive state updates increase the bit lengths of common denominators
additively. Fraction-free arithmetic gives polynomial bit complexity and $O(Np^2)$ field operations.

For scaling, form $F=\sum_i\Tr(A_i)A_i$ in $O(Nd^2)$ operations. If $F\ne0$, then $\Tr(F)/d\le V\le\Tr(F)$. Bisection with rational positive-semidefiniteness tests finds $V\le v\le(1+\xi)V$ in $O(\log(d/\xi))$ tests, each costing $O(d^3)$. Take $\xi=4\cdot10^{-7}$ and bracket $\sqrt v$ by a rational $b$ with $v\le b^2\le(1+\xi)v$. This proves the stated scale cost (counting scalar square roots as elementary operations) and the required relative error. Rational root bracketing has polynomial bit cost. For the limiting-constant variant, add $O(\log(1/\epsilon))$ bisections.

\subsection{The inverse recursion}\label{sec:oracle}
We prove the gap and approximation estimates for~\eqref{eq:inverseiteration}.
The notation $\beta,F,X_t,Y_t$ is as in Section~\ref{sec:complexity}.
Every feasible pair dominates the monotone iterates; their limit is therefore
the least feasible pair.

\begin{lemma}[Gap and root-resolvent approximation]\label{lem:oracle}
At an accepted state the minimizer has $t_*<4$ and
$t_*-\beta\ge g:=\rho^2/2048$.
For $\delta=t-\beta>0$, $\beta>0$, and
$q=(t-\sqrt{t^2-\beta^2})/\beta$, the iterates satisfy
\begin{equation}\label{eq:chebtails}
 \begin{aligned}
 \norm{X_t-X_h},\norm{Y_t-Y_h}&\le4\delta^{-1}q^{2h+2},\\
 \norm{U_t-U_h},\norm{V_t-V_h}&\le4(2h+4)\delta^{-2}q^{2h+2},
 \end{aligned}
\end{equation}
where $U_h=-\partial_tX_h$, $V_h=-\partial_tY_h$. Also
$F''(t)\ge4\rho d/(t+\beta)^3$. If $\beta=0$, depth zero is exact.
\end{lemma}
\begin{proof}
For the gap, perturb the optimum to $(X+aP,Y+aQ)$. Since $\norm{X^{-1}},\norm{Y^{-1}}\le8$ and $P,Q\preceq I$,
\[
 (X+aP)^{-1}\preceq X^{-1}-aX^{-1}PX^{-1}+512a^2I.
\]
Stationarity cancels the first-order contribution of the trace map. Both constraints are then
feasible at level $t_*-a\rho+512a^2$. Taking $a=\rho/1024$ proves the gap.

To prove the tails, express the trace map as a sum $\Eop(Z)=\sum_jK_jZK_j^*$ with Hermitian Kraus matrices. For example take $\sqrt{c\psi_i}M_i^{1/2}H_\alpha M_i^{1/2}$, where the $H_\alpha$ form a Hermitian orthonormal basis. Build a rooted tree with $d$-dimensional fibers, alternating diagonal blocks $S,-S$, and edge blocks $K_j$. Schur complementation identifies~\eqref{eq:inverseiteration} with root blocks of its height-$h$ resolvents. The two parity operators are negatives of unitarily conjugate operators. Their common norm is $\beta$. To see this, feasible pairs bound all finite
truncations. Above the operator norms, the monotone root resolvents converge
to a feasible pair. Finite-support vectors show that the truncation norms have supremum $\beta$.

A closed walk from the root needs at least $2h+2$ edges to leave the truncation. Thus the finite and infinite root moments agree through degree $2h+1$. On $[-\beta,\beta]$ the Chebyshev expansion is
\[
 \frac1{t-z}=\frac1{\sqrt{t^2-\beta^2}}
 \left(1+2\sum_{k\ge1}q^kT_k(z/\beta)\right).
\]
Truncate at degree $2h+1$ and bound both compression errors. Since
$\sqrt{t^2-\beta^2}(1-q)=\delta(1+q)\ge\delta$, this gives the first bound in~\eqref{eq:chebtails}. Differentiate the tail, using $q'=-q/\sqrt{t^2-\beta^2}$, to obtain the second. Finally, the second derivative of a resolvent is twice its cube. At this
order, both root compressions are bounded below by $(t+\beta)^{-3}I$, which
proves the convexity estimate.
\end{proof}

The tree and Chebyshev polynomials are used only in the proof. The algorithm performs the elementary recursion~\eqref{eq:inverseiteration} and its derivative recursion
$U_{h+1}=X_{h+1}(I+\Eop(V_h))X_{h+1}$, with $U_0=X_0^2$, and the analogous formula for $V_h$.
Let $F_h=t+\rho\Tr(X_h+Y_h)$ and
$\alpha_h=\max(\norm{\Eop(X_h-X_{h-1})}_F,\norm{\Eop(Y_h-Y_{h-1})}_F)$.
The finite pair is feasible after raising $t$ by $\alpha_h$. Convexity gives the value interval
\begin{equation}\label{eq:valueinterval}
 F_h(t)-4|F_h'(t)|\le R(x)\le F_h(t)+\alpha_h,
 \qquad 0<t\le4.
\end{equation}
The lower inequality follows by evaluating the tangent line at the true minimizer $t_*\in(0,4)$.

Since $-\log q\ge\sqrt{\delta/t}$ and $g^{-1}=O(d^2)$, depth
$h=O(d\Lambda)$ gives any prescribed accuracy $\eta$ for the value and data.
Here $\Lambda=1+\log d+\log(1/\eta)$, and the internal tolerances are
polynomially smaller. For the scalar search, at $t=\beta+g/2$ the derivative
of $F$ is negative and bounded away from zero by a polynomial multiple of
$\rho dg$. The tail bounds transfer this sign to $F_h'$, so its minimizer
stays above this level. Guarded bisection uses $O(\Lambda)$ probes. If a Schur denominator is smaller than $g/16$, the queried level is below
the finite-tree norm plus $g/16$. It is therefore below the minimizer and
provides a lower bracket. All inverted denominators otherwise have a known positive bound. The derivative recursion and $F''$ bound recover the optimizer and $P=\rho U_t$, $Q=\rho V_t$.

\subsection{Response solves and endpoint tests}\label{sec:linearsolves}
With the current optimizer computed, we turn to the cached solves and
endpoint witnesses used in lines~\ref{line:endpoint}, \ref{line:frame},
and~\ref{line:trials}.
Cache $XM_iX$, $YM_iY$, and their trace Gramians. A solve $\Lop(U,V)=(F,G)$ reduces to the $2m\times2m$ coefficient system for $u_i=\Tr(M_iU)$, $v_i=\Tr(M_iV)$:
\[
 \begin{pmatrix}I&-cG_X\diag(\psi)\\-cG_Y\diag(\psi)&I\end{pmatrix}
 \binom u v=\binom{(\Tr(M_iXFX))_i}{(\Tr(M_iYGY))_i},
 \qquad (G_X)_{ij}=\Tr(M_iXM_jX).
\]
Recover $U=XFX+c\sum_i\psi_i v_iXM_iX$ and its companion. Preparation costs
$C_F(m,d)=O(md^3+m^2d^2+m^3)$; a further right side costs
$C_S(m,d)=O(md^2+m^2+d^3)$.

After this preparation, a solve with the KKT Jacobian~\eqref{eq:jacobian}
has the same order of cost. Cache
$B=\Lop^{-1}\mathbf e$, $D_1=\Lop^{-1}\mathscr HB$, and $s_0=\ip{\mathbf e}{D_1}\ge1/(d\sqrt\rho)$.
For new data $f_1,f_2,f_3$, set
\[
 A=\Lop^{-1}f_1,\quad C=\Lop^{-1}(f_2+\mathscr HA),\quad
 \lambda=(\ip{\mathbf e}{C}-f_3)/s_0,\quad U=A-\lambda B,\quad V=C-\lambda D_1.
\]
This uses two response solves, matrix products, and one scalar division.

Endpoint witnesses are cheaper: using $\norm{M_i}\le1$, their upper increases are
$[1-x_i-c\psi_i\Tr(M_iY)]_+$ and
$[1+x_i-c\psi_i\Tr(M_iX)]_+$.
Allocate half of $\delta_E$ to this test and half to current-data error. Every endpoint that preserves feasibility with the exact current data then
passes after refinement to polynomial precision. Failure therefore implies strict lightness. Neither an endpoint
test nor a local trial requires a new optimization.

\subsection{Uniform precision and progress}\label{sec:precision}
We fix the precision and sufficient scale for the acceptance test
\eqref{eq:acceptance} in line~\ref{line:accept}. Work on the sublevel $R\le K=16$ and the face $|x_i|\le1-\sigma/2$. The derivatives of the weight function through order three are bounded by fixed polynomials in $\sigma^{-1}$, directly from~\eqref{eq:profile}. On a light face one may use
\begin{equation}\label{eq:scalebounds}
 t_i^2\ge1/4,\quad t_i\le4/\sigma,\quad
 \ell_i,k_i\ge\sigma/10,\quad \ell_i^2,k_i^2\le16/\sigma.
\end{equation}
The sublevel bounds give $a_i,b_i\ge\tau/(2K)$ and $p_i,q_i\ge\rho\tau/(4K^2)$. Hence
$(D_F)_{ii},(D_G)_{ii}\ge F_-:=\sigma^2\rho\tau^2/(800K^3)$; their upper bounds and those of $z_i,w_i,\widetilde a_i,\widetilde b_i$ are polynomial. Equations~\eqref{eq:certificatebase}--\eqref{eq:normalizedmetric} then give polynomial bounds for $W,W^{-1}$ and $L_i$, and inverse-polynomial lower bounds for $g_0(x)$ and $a_x$. These quantities and~\eqref{eq:shape} bound all denominators away from zero.

Two cancellations give useful bounds on the response frame. With $S_D=\diag(D_F^{1/2},D_G^{1/2})$ and $\mathsf T_0=\left(\begin{smallmatrix}0&\mathsf B\\\mathsf A&0\end{smallmatrix}\right)$,
\begin{equation}\label{eq:dictionarycancellation}
 \Bcal=W^{-1}S_D(I-\mathsf T_0)^{-1}\binom I{-I}\diag(t_i),
 \qquad \norm{(I-\mathsf T_0)^{-1}}\le1+\frac{16c}{\rho\sigma}.
\end{equation}
In this identity, the diagonal changes of scale cancel before we take norms. For the inverse bound, apply~\eqref{eq:tracecontraction} to the maps $v\mapsto\sum_i\ell_iv_iM_i$ and $v\mapsto\sum_i k_iv_iM_i$, then use $\norm{\Lop^{-1}}\le\rho^{-1}$ as in~\eqref{eq:coeffresponse}. Also lightness and $\psi_i\ge1-x_i^2$ give $C_i\ge c$; thus~\eqref{eq:purities} yields $\norm{\mathsf T}_F^2\le2m/c$.

For explicit bounds, let $F_+$ bound $D_F,D_G$ above, $D_-=2F_-$, $t_-=1/2$, $t_+=4/\sigma$, and $R_+=1+16c/(\rho\sigma)$. Then
\begin{equation}\label{eq:grambounds}
 \Gamma_-I\preceq\Gamma\preceq\Gamma_+I,\qquad
 \Gamma_-:=\frac{t_-^2D_-}{2\norm W^2(1+2m/c)},\quad
 \Gamma_+:=2\norm{W^{-1}}^2F_+R_+^2t_+^2,
 \qquad \Sigma\succeq\frac{\Gamma_-}{m\Gamma_+}I.
\end{equation}
Replacing the norms by their uniform rational bounds gives $\lambda_*$.
Normalizing the covariance loses only one factor of $m$.

For derivatives,~\eqref{eq:jacobian} bounds the inverse KKT Jacobian by a polynomial $K_J$. Direct inverse differentiation and~\eqref{eq:tracecontraction} bound the joint derivatives of $\mathcal F$ through order three by a polynomial $A_0$ on a product neighborhood. For example, the bound $A_0=10^{12}(1+d)\rho^{-1}\sigma^{-4}$ is larger
than needed. The inverse terms contain at most five inverse factors.
Equation~\eqref{eq:tracecontraction} bounds the derivatives of the trace map,
and all remaining factors from the product rule are fixed. Implicit differentiation gives
\[
 b_1=K_JA_0,\quad b_2=K_JA_0(1+b_1)^2,\quad
 b_3=K_JA_0\bigl((1+b_1)^3+3(1+b_1)b_2\bigr).
\]
Multiplying by the norm of the objective's linear functional bounds the derivatives of $R$. The same calculation bounds the cubic constant $C$ and radius $d_0$ in Lemma~\ref{lem:witness}. Taking $H\ge1+\sqrt M$, $b_*=g_*\sqrt{\lambda_*}$, and an upper bound $\bar a\ge a_x$, a sufficient scale is
\begin{equation}\label{eq:stepfloor}
 s_c=\min\left\{\frac{d_0}{H},\frac{a_*}{4CH^3},
 \frac{b_*}{4(M_2H^2+\bar a)}\right\}.
\end{equation}
Decrease $d_0$ if necessary to keep all candidates within the truncated face. The paired Taylor argument of Section~\ref{sec:paired} now applies to feasible witnesses, with a decrease at least $a_xs^2/2$ for some candidate at every $s\le s_c$.

Use rational interval data and a fixed dyadic coordinate grid. Choose the
errors in the current value, directions, and grid so that the resulting
witness error is at most $a_*s_c^2/1024$. All maps used above have polynomial condition bounds, so inverse-polynomial input accuracies suffice. With $\widehat a(x)\in[a_x/2,a_x]$, use the acceptance test
\eqref{eq:acceptance}; these error bounds account for grid rounding in
line~\ref{line:trials}. Some candidate passes by scale $s_c/2$; every accepted local move decreases $R$ by at least
$\Delta_*=a_*s_c^2/128$.
Thus $L\le4/\Delta_*$ and $J_{\max}\le\lceil\log_2(2/s_c)\rceil+1$. These bounds prove the polynomial progress assertions.

Finite inversion recursions have polynomially many layers and denominators bounded away from zero. Although interval sensitivities may grow exponentially with the number of
layers, their logarithms remain polynomial. Polynomially many working bits
therefore suffice.
By~\eqref{eq:grambounds}, the same holds for coefficient solves and Cholesky
factorization. Frozen coordinates are stored exactly; active ones are rounded to the fixed grid. Together with rational preprocessing and norm scaling, this proves polynomial bit complexity in the input encoding. For the $\Czero+\epsilon$ variant, the same bounds are polynomial in $\epsilon^{-1}$.

\subsection{A matrix-free variant}
A randomized alternative constructs the response frame without forming a
dense matrix. A product by $\mathsf A$ or $\mathsf B$ costs $O(md^2+d^3)$ via
$\mathsf A u=c\diag(\ell_i)\mathcal A^*(X\mathcal A(\diag(\ell_i)u)X)$.
Let $J_\perp$ complete $J$ orthogonally and set $K_r=J_\perp^*(I-\mathsf T)W$. Projecting a random response onto $\ker K_r$ requires a solve with $K_rK_r^*$. At direction tolerance $\eta_h$, conjugate gradients costs
$O(\kappa_2(K_r)(md^2+d^3)\log(2/\eta_h))$ operations per probe, with polynomial conditioning and certified residual stopping. This construction avoids forming a dense Gram matrix, but its cost depends
on the displayed condition number. Its progress bound also requires an
additional normalization of the covariance. Theorem~\ref{thm:complexity} concerns the deterministic
frame construction.

\section{Further details of the applications}\label{sec:applicationdetails}
\subsection{Repeated halving}\label{sec:halvingproof}
We prove Proposition~\ref{prop:halving}. The half-partition error is below $1.684\sqrt{\varepsilon\norm{T_S}}$. In Weaver's normalization,
$\sum_iv_iv_i^*=16I$ and $\norm{v_i}^2\le1$, so each part has norm below
$8+4(1.684)=16-1.264$.

\begin{proof}
Use the recurrence~\eqref{eq:halvingrecurrence} from the main proof.
If $e_j\le C\sqrt{\varepsilon t_0}2^{-j/2}$, the next normalized error is at most
$C/\sqrt2+1.684\sqrt2\sqrt{1+C\sqrt{2^j\varepsilon/t_0}}$.
For $C=11.5$ and $j<k$, this is at most
$11.5/\sqrt2+1.684\sqrt2\sqrt{1+11.5/\sqrt{134}}<11.495<11.5$.
The numerical inequality is verified by rational square comparisons in the supplementary programs~\cite{certificate}. Empty or zero parts split trivially. The bound $2^k\le t_0/(67\varepsilon)\le N/67$ makes the output size and call count polynomial.
\end{proof}
\subsection{Paving and rational square roots}\label{sec:pavingproof}
We use the repeated-halving bound to prove Proposition~\ref{prop:paving}.
\begin{proof}
For factored input apply Proposition~\ref{prop:halving} to $A_i=(Ce_i)(Ce_i)^*$, whose traces are at most $\gamma$. Their sum is $CC^*$, of norm $t_0=\norm T$, and $\norm{\sum_{i\in S}A_i}=\norm{P_STP_S}$.
If $t_0\le\epsilon$ take one part. Otherwise enlarge the diagonal bound to $\gamma=\epsilon/268$ and take the largest $k$ with $2^k\le4t_0/\epsilon=t_0/(67\gamma)$. Then $2^{-k}t_0<134\gamma$, and~\eqref{eq:halving} gives
\[
 \norm{P_STP_S}<(134+11.5\sqrt{134})\gamma<268\gamma=\epsilon,
 \qquad r=2^k\le4t_0/\epsilon\le4/\epsilon.
\]
All comparisons involving $t_0$ are rational definiteness tests.

For unfactored input enlarge $\gamma$ to $\epsilon/269$. Obtain a rational Hermitian $C$ with $\norm{C-T^{1/2}}\le\gamma/10^4$, as justified below. Then $\widehat T=C^2$ has $\norm{\widehat T-T}\le3\gamma/10^4$, diagonal bound $\widehat\gamma=(1+3\cdot10^{-4})\gamma$, and norm $\widehat t_0\le1+3\gamma/10^4$. Unless the original matrix already has norm at most $\epsilon$, take the largest $k$ with $2^k\le\widehat t_0/(67\widehat\gamma)$. The same estimate, which does not require $\widehat t_0\le1$, gives
$\norm{P_STP_S}<268\widehat\gamma+3\gamma/10^4<269\gamma$.
Also $r\le\widehat t_0/(67\widehat\gamma)\le1/(67\gamma)=269/(67\epsilon)<4.02/\epsilon$.
The number of parts is at most a constant times the matrix dimension, since $\widehat t_0\le\Tr\widehat T\le d\widehat\gamma$.
\end{proof}

\paragraph{Rational square roots to prescribed accuracy.}
The unfactored case uses the following approximation: a rational PSD contraction has a rational Hermitian square-root approximation to error $\delta$ in time polynomial in $d$, the input length, and $\log(1/\delta)$. For $0<\delta<1$, set $a=\delta/4$, $H=T+a^2I$, and start Newton iteration $C_0=2I$, $C_{j+1}=(C_j+HC_j^{-1})/2$. The exact iterates commute with $H$. For each eigenvalue $\lambda\in[a^2,1+a^2]$, the ratio $(c_j-\sqrt\lambda)/(c_j+\sqrt\lambda)$ squares at every step. Thus $O(\log(1/a)+\log(1+\log(1/\delta)))$ iterations suffice for error $\delta/4$ relative to $H^{1/2}$. Symmetrize every computed iterate and round it to a prescribed Frobenius accuracy. As long as the accumulated error is below $a/2$, inverses have norm at most $2/a$, and a step amplifies error by at most $K_a=1+4/a^2$. A per-iteration rounding error below $\delta/(16kK_a^k)$ keeps accumulated error below $\delta/16<a/2$ for all $k$ iterations. With polynomial bit precision, we can use exact rational inversion followed
by this rounding. Finally $\norm{H^{1/2}-T^{1/2}}\le a$. This removes any factorization oracle from the second paving assertion.


\begin{thebibliography}{99}
\bibitem{BL} Y.~Bilu and N.~Linial.
Lifts, discrepancy and nearly optimal spectral gap.
\emph{Combinatorica} \textbf{26} (2006), 495--519.
\href{https://arxiv.org/abs/math/0312022}{arXiv:math/0312022}.
\bibitem{EJ} E.~Ezeunala and H.~Jiang.
Rank-one matrix discrepancy and algorithmic Kadison--Singer.
\href{https://arxiv.org/abs/2609.17266}{arXiv:2609.17266}, 2026.
\bibitem{KLS} R.~Kyng, K.~Luh, and Z.~Song.
Four deviations suffice for rank 1 matrices.
\emph{Advances in Mathematics} \textbf{375} (2020).
\href{https://arxiv.org/abs/1901.06731}{arXiv:1901.06731}.
\bibitem{MSS1} A.~W.~Marcus, D.~A.~Spielman, and N.~Srivastava.
Interlacing families I: Bipartite Ramanujan graphs of all degrees.
\emph{Annals of Mathematics} \textbf{182} (2015), 307--325.
\href{https://arxiv.org/abs/1304.4132}{arXiv:1304.4132}.
\bibitem{MSS2} A.~W.~Marcus, D.~A.~Spielman, and N.~Srivastava.
Interlacing families II: Mixed characteristic polynomials and the Kadison--Singer problem.
\emph{Annals of Mathematics} \textbf{182} (2015), 327--350.
\href{https://arxiv.org/abs/1306.3969}{arXiv:1306.3969}.
\bibitem{MSS4} A.~W.~Marcus, D.~A.~Spielman, and N.~Srivastava.
Interlacing families IV: Bipartite Ramanujan graphs of all sizes.
\emph{SIAM Journal on Computing} \textbf{47} (2018), 2488--2509.
\href{https://arxiv.org/abs/1505.08010}{arXiv:1505.08010}.
\bibitem{Cohen} M.~B.~Cohen.
Ramanujan graphs in polynomial time.
\emph{Proceedings of FOCS} (2016).
\href{https://arxiv.org/abs/1604.03544}{arXiv:1604.03544}.
\bibitem{MoserTardos} R.~A.~Moser and G.~Tardos.
A constructive proof of the general Lov\'asz Local Lemma.
\emph{Journal of the ACM} \textbf{57}(2), Article 11 (2010).
\href{https://arxiv.org/abs/0903.0544}{arXiv:0903.0544}.
\bibitem{Nilli} A.~Nilli.
On the second eigenvalue of a graph.
\emph{Discrete Mathematics} \textbf{91} (1991), 207--210.
\bibitem{Weaver} N.~Weaver.
The Kadison--Singer problem in discrepancy theory.
\emph{Discrete Mathematics} \textbf{278} (2004), 227--239.
\bibitem{BSB} E.~Akbas and S.~Sra.
\emph{Boolean small-ball inequalities for discrepancy theory}.
September 2026.
\bibitem{XXZ} J.~Xie, Z.~Xu, and Z.~Zhu.
Upper and lower bounds for matrix discrepancy.
\emph{Journal of Fourier Analysis and Applications} \textbf{28} (2022), Article 81.
\href{https://arxiv.org/abs/2006.12083}{arXiv:2006.12083}.
\bibitem{SongYue} Z.~Song and S.~Yue.
Polynomial time algorithms for the Kadison--Singer problem.
\href{https://arxiv.org/abs/2609.19794}{arXiv:2609.19794}, 2026.
\bibitem{KathuriaKS} T.~Kathuria.
A walk from free probability to matrix discrepancy II: Weaver's problem and the Kadison--Singer conjecture.
\href{https://arxiv.org/abs/2609.18913}{arXiv:2609.18913}, 2026.
\bibitem{RS} M.~Ravichandran and N.~Srivastava.
Asymptotically optimal multi-paving.
\emph{International Mathematics Research Notices} \textbf{2021}(14), 10908--10940.
\href{https://arxiv.org/abs/1706.03737}{arXiv:1706.03737}.
\bibitem{certificate} A.~Jadbabaie, A.~Saberi, and S.~Sra.
\emph{Exact scalar certificates and implementation analyses for the Gram-response rounding method}.
Supplementary code and certificates, 2026. To be made available on GitHub.
Exact parameters, verifiers, subdivision certificates, and verification records,
including exhaustive local repair checks, prefix checks on $K_{2,2}$, the
one-sided bipartite implementation, and the rational two-sided implementation.
\end{thebibliography}
\end{document}